\documentclass[a4paper,11pt]{article}

\usepackage[english]{babel}
\usepackage[utf8]{inputenc}
\usepackage[OT1]{fontenc}

\usepackage{fullpage}
\usepackage{setspace}
\usepackage{amsmath, amssymb, amsthm}
\newtheorem{lemma}{Lemma}
\newtheorem{proposition}{Proposition}
\theoremstyle{remark}
\newtheorem{rmk}{Remark}
\usepackage{natbib}
\usepackage{graphicx}

\usepackage{url}

\everymath{\displaystyle}

\title{Exact and asymptotic solutions of the call auction problem}
\author{Ioane Muni Toke
\\ University of New Caledonia -- ERIM, Noumea, New Caledonia,
\\ Ecole Centrale Paris -- Chair of Quantitative Finance, Paris, France.}
\date{This version: \today}

\begin{document}
\maketitle

\begin{abstract}
The call auction is a widely used trading mechanism, especially during the opening and closing periods of financial markets. 
In this paper, we study a standard call auction problem where orders are submitted according to Poisson processes, with random prices distributed according to a general distribution $F$, and may be cancelled at any time.
We compute the analytical expressions of the distributions of the traded volume, of the lower and upper bounds of the clearing prices, and of the price range of these possible clearing prices of the call auction.
Using results from the theory of order statistics and a theorem on the limit of sequences of random variables with independent random indices, we derive the weak limits of all these distributions.
In this setting, traded volume and bounds of the clearing prices are found to be asymptotically normal, while the clearing price range is asymptotically exponential. All the parameters of these distributions are explicitly derived as functions of the parameters of the incoming orders' flows.
\end{abstract}

\section{Introduction}

During the past decades, financial markets have progressively adopted forms of continuous double auctions as the main structure of exchange.
In a continuous double auction, traders can submit trading orders at any time, and these orders are immediately treated. A centralized structure, the order book, stores the list of non-filled orders. Any buy (resp. sell) order submitted with a limit price lower (resp. higher) than the best ask (resp. bid) price is stored in the order book (limit order). Any  buy (resp. sell) order submitted with a price higher (resp. lower) than the best ask (resp. bid) is executed (filled), and the order book is modified accordingly (market order or crossing limit order). Finally, any limit order waiting in the order book can be cancelled at any time.
As of today, major trading systems use the continuous double auction at least during the main trading day period. 
As a consequence, microstructure theory has a large and growing body of literature focusing on the continuous double auction and order book modelling, both empirical \citep[see e.g.][among others]{Biais1995,Challet2001,Bouchaud2002,Potters2003,Mike2008} and theoretical \citep[see e.g.][]{Smith2003,Preis2006,Cont2010,Abergel2013,MuniToke2014}. See also \citet{Bouchaud2009,Chakraborti2011} for some review material.

The choice of the continuous double auction was mainly done during the computerization of the financial stock exchanges all over the world, starting in the 1980's. Such a choice was however not obvious, and another trading mechanism, the periodic call auction, has had its supporters.
In a periodic call auction, trading orders can be submitted continuously but are filled only periodically, using some clearing mechanism. During the auction time, all submitted orders are stored (and may be cancelled at any time), whatever their submission price is. Then at some pre-determined time the auction is closed, and the cumulative bid and offer are computed as functions of the price. The clearing price is chosen upon consideration of these functions, following a set of pre-agreed rules that aim at maximizing the trading volume, minimizing the market imbalance, etc. \citep[see e.g.][for examples of clearing rules]{ComertonForde2006,ComertonForde2006b}. Once a clearing price is computed, all exchanges are done all at the same time. 
Proponents of the call auction mechanism argued that the call auction might help reducing trading cost and enhancing price discovery \citep[see e.g.][]{Pagano2003,Schwartz2003}, and the call auction has been used as  the main trading mechanism of the Arizona Stock Exchange that has been operational from 1990 to 2001. 
Note also that recent trading accidents due to uncontrolled algorithmic trading have raised some new interests for the call auction: its non-continuous periodic trading times may appear as a potential solution against "flash crashes" such as the one observed on the NYSE on May 6th, 2010.

As of today, the call auction is used in different trading procedures, according to rules that often depend on the liquidity of the traded stocks.
For example, the call auction is the main trading mechanism for the least traded equities on the Euronext French regulated market. These stocks are not continuously traded but are actually fully exchanged using call auctions that take place once a day for some stocks, and twice a day for most of them, at 11:30 a.m. and 4:30 p.m. \citep{Euronext2013}.
A better known example is the use of the call auction to determine opening and closing prices for liquid stocks, that are otherwise traded according to a continuous double auction throughout the trading day.
For example, the typical trading day for continuously-traded liquid stocks on the Euronext French regulated market is as follows. The trading day starts with an opening call auction until 9:00 a.m.: all orders submitted before this time are "automatically recorded in the Central Order Book without giving rise to trades". Then at 9:00 a.m., the clearing of the call auction takes place, which determines the opening price. From 9:00 a.m. to 5:30 p.m., trading is done according to the continuous double auction.
At 5:30 p.m. starts the closing call auction. During 5 minutes, submitted orders are again "automatically recorded in the Central Order Book without giving rise to trades". At 5:35 p.m., the market is cleared, and this closing auction determines the closing price. The trading day ends with five minutes of "Trading at last" in which trades can be made at the closing price \citep[see][for more details]{Euronext2013,Euronext2014}.

Despite the wide interest in the call auction, microstructure literature on this mechanism is scarce. The seminal work in this field is due to \citet{Mendelson1982}, in which the call auction described here is called a "clearing house" mechanism.
These results were later used for a comparison between call auction and continuous double auction by \citet{DomowitzWang1994}, and the paper enjoys an increasing trend in citations in recent microstructure papers.
\citet{Mendelson1982} models a call auction in which buy and sell orders are submitted with size one and a price uniformly distributed on some interval $(0,m)$. In the case where buy and sell orders are assumed to be Poisson processes with \emph{identical} parameters, the author derives the distribution of the traded volume
as well as the first moments of the price distribution. However, in the general case with potential market imbalance (i.e. with different processes governing the arrival of buy and sell orders), weaker results are only obtained through the use of asymptotic renewal theory results, in the form of second-order expansions of the first and second moments of the traded volume and the price distributions. 
The main assumption required for these approximations to be acceptable is that the market must be "thick" in Mendelson's words, i.e. very liquid. Therefore the results mentionned above are not valid in the case of an illiquid market where the probability that no transaction occurs during a call auction is not negligible.

In this paper, we provide the complete solution of the call auction described in \cite{Mendelson1982}, even in the case where this seminal paper stated that such general forms did not exist. We show that the problem of determining the traded volume and the price distributions is analytically tractable, even in the general case of market imbalance and a general price distribution. Perhaps more importantly, we derive rigorous weak limits of the exact distributions obtained, showing the asymptotic normality of the traded volume and the price distributions as the liquidity increases. Similarly, it is proved that in this setting the clearing price range, a proxy to the post-clearing spread, is asymptotically exponentially distributed. 

The remainder of the paper is organized as follows. In section \ref{sec:PreliminaryResults}, the call auction is described and it is shown that, conditionally on the number of submitted bid and ask orders, the traded volume and price distribution are easily analytically tractable. In section \ref{sec:TradedVolumeDistribution}, the exact unconditional distribution of the traded volume is derived, as well as its weak limit when the liquidity increases. It is shown that the orders' price distribution does not influence this distribution, and that the market imbalance (ratio between the expected number of bid and ask orders) has a non trivial influence on the limit distribution. Using in particular fundamental results of the theory of order statistics, section \ref{sec:PricesDistributions} derives similar results for the distributions of the lower and upper clearing price, and section \ref{sec:PriceRangeDistribution} for the distribution of the range of potential clearing prices. The scaling constants are again non-trivial and depend on the market imbalance as well as the orders' price distribution.
As all these results are derived without any cancellation mechanism, section \ref{sec:Cancellation} shows that, thanks to a simple substitution in the model parameters, all results can easily be generalized to the case where submitted limit orders are cancelled after an exponentially distributed lifetime.

To our knowledge, all these results are new.

\textbf{Notations:} We will use the following (usual) notations in the remainder of the paper: 
$\mathcal N(m,\sigma)$ denotes the normal distribution with mean $m$ and standard deviation $\sigma$ ; 
$\Phi$ denotes the cumulative distribution function of $\mathcal N(0,1)$ ; 
$\mathcal B(n,p)$ denotes the binomial distribution with $n$ trials and success probability $p$ ; 
$\mathcal E(\lambda)$ denotes the exponential distribution with parameter $\lambda$ (inverse of the mean) ;
$\overset{d}{\underset{n\to+\infty}\longrightarrow}$ denotes the weak convergence as $n$ tends to infinity \citep[see e.g.][]{Billingsley1999} ; 
$\lfloor x\rfloor$ denotes the integer part of a real number $x$ ; $\triangleq$ will be used for the definition equality.
Finally, all random variables and processes subsequently used in this paper are assumed to be defined on some probability space $(\Omega,\mathcal F,\mathbf P)$.

\section{Description of the call auction and preliminary results}
\label{sec:PreliminaryResults}

Let us consider a standard call auction for the exchange of a given financial product.
Ask (sell) orders are submitted at random times according to a Poisson process with parameter $\lambda_A$. We assume that all ask orders are unit-sized, and that their prices form a set of independent random variables identically distributed according to some distribution $F$.
Similarly, bid (buy) orders are submitted at random times according to an independent Poisson process with parameter $\lambda_B$. It is assumed as well that all bid orders are unit-sized, and that their prices form a set of independent random variables identically distributed according to the distribution $F$. 
The assumption that bid and ask orders are submitted according to the same distribution may seem very basic, but it is the one used by \citet{Mendelson1982} and it is actually fundamental to the following analysis (see Remark \ref{rmk:SameBuySellPriceDistribution}).
It is a purely zero-intelligence assumption, where no prior information on the price is incorporated into the model.
Note however that we have no restriction on $F$, which consequently can potentially accomodate a wide range of empirical fits. In any case, without further empirical investigations on the placement of order during a call auction, it should not be discarded, given its analytical potential.

Market imbalance, which represents the relative shares of bid and ask submissions, is thus described by the parameter $\alpha \triangleq \frac{\lambda_A}{\lambda}$ where $\lambda\triangleq\lambda_A+\lambda_B$ is the total rate of orders submission. Obviously, the market is symmetric when $\alpha=\frac{1}{2}$.

The call auction opens (i.e. starts accepting order submission) at time $0$ and closes at a deterministic time $T>0$.
For the ease of exposition and computation, no cancellation of a submitted trading order is allowed for the moment, but this restriction will easily be lifted in section \ref{sec:Cancellation}.
Once the call auction is closed, all submitted orders are taken into account and the market is cleared, i.e. a clearing price (or exchange price, or trading price in the literature) is decided.
There are several types of rules used in practice to determine the clearing price \citep[see e.g.][Table 1]{ComertonForde2006}, most of which start with maximizing the traded volume. For a given clearing price $p$, the traded volume is defined as the quantity of shares that can be matched at that price, ans is equal to the minimum of the total ask quantity offered below $p$ and the total bid quantity offered above $p$.
There is not necessarily a unique price that maximizes the traded volume, and more rules may be used in practice to determine a unique price. In this study however, one of our main output will be the range of possible clearing prices maximizing the traded volume, and we will not use further rules to determine a unique price.

Let us rephrase this description mathematically. Let $A(p)$ be the cumulative number of ask orders up to price $p$ at time $T$. $p\mapsto A(p)$ is $\mathbf P$-almost surely a positive non-decreasing right-continuous step function, with unit-size steps.
Similarly, $B(p)$ is the cumulative number of bid orders \emph{down to} price $p$, and $p\mapsto B(p)$ is $\mathbf P$-almost surely a positive non-increasing left-continuous step function, with unit-size steps.
Conditionally on $A(\infty)\triangleq\lim_{p\to+\infty} A(p)\geq 1$ and $B(-\infty)\triangleq\lim_{p\to-\infty} B(p)\geq 1$, i.e. assuming that there is at least one buy order and one sell order, $p$ is a clearing price maximizing the traded volume $V$ if and only if $A(p)=B(p)$.
Since all orders are unit-sized, the clearing price maximizing the traded volume is $\mathbf P$-almost surely not unique. Furthermore, the monotony of $A$ and $B$ ensures that the set of possible exchange prices is an interval $(L,U)$. A detailed illustration is provided on figure \ref{fig:CallAuctionPrinciple} (some of the notations used there will be introduced later in the text).
\begin{figure}
\centering
\includegraphics[width=0.8\textwidth]{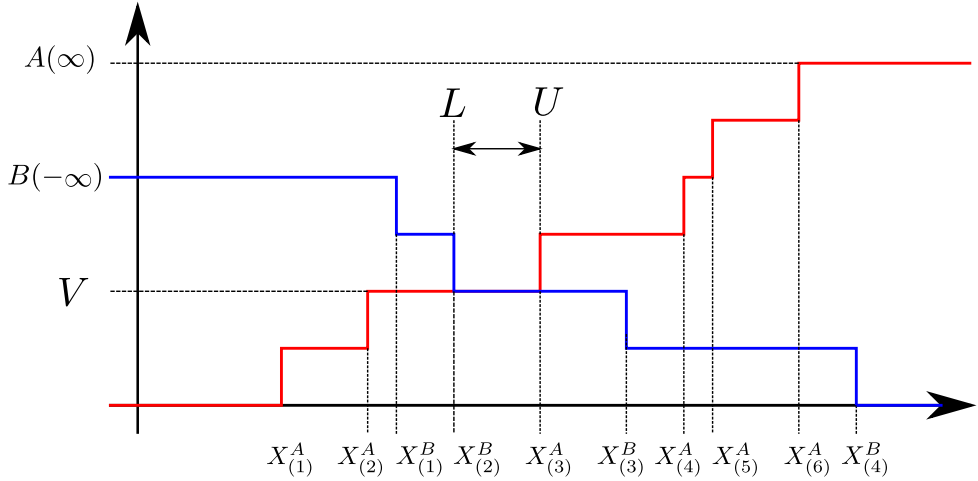}
\label{fig:CallAuctionPrinciple}
\caption{Example of the clearing mechanism ending a call auction, with price on the $x$-axis and volumes on the $y$-axis. The cumulative number of ask orders $A(p)$ is in red, while the cumulative number of bid orders $B(p)$ is in blue. $V$ is the traded volume and $(L,U)$ is the range of possible clearing prices.}
\end{figure}

The basic finding that helps us solving the general call auction problem is that, conditionally on the random variables $A(\infty)$ and $B(-\infty)$, which denote the number of ask and bid orders submitted upon the closing of the call auction, the system is actually easily tractable using the framework of order statistics, as stated in the following lemma.

\begin{lemma}
\label{lem:ConditionalDistributions}
Let $(m,n)\in(\mathbb N^*)^2$.
Conditionally on the set $\left\{A(\infty)=m,B(-\infty)=n\right\}$, 
\begin{enumerate}
	\item the traded volume $V$ is hypergeometrically distributed:
\begin{equation}
	\label{eq:ConditionalTradedVolume}
	\forall k\in\{0,\ldots,\min(n,m)\},\; 
	\mathbf P(V=k \vert A(\infty)=m,B(-\infty)=n) = \frac{\binom{m}{k}\binom{n}{k}}{\binom{m+n}{n}},
\end{equation}
and the distribution of the traded volume does not depend on the price distribution $F$ ;
	\item the lowest possible clearing prices $L$ maximizing the traded volume is distributed as the $n$-th order statistics of a random sample with size $n+m$ distributed with distribution function $F$ ;
	\item the highest possible clearing prices $U$ maximizing the traded volume is distributed as the $(n+1)$-th order statistics of a random sample with size $n+m$ distributed with distribution function $F$.
\end{enumerate}
\end{lemma}
\begin{proof}
Assume that the call auction is closed and that $A(\infty)=m$ and $B(-\infty)=n$. 
All $n+m$ orders form a set of $n+m$ (independent and identically distributed) points on $\mathbb R$. Let $x_1<\ldots<x_{n+m}$ be this increasing finite sequence of points on the real line. Let us consider this set of points given and decide whether a given point $x_i$ is an ask order ($m$ points should be "ask" in the whole set) or a bid order ($n$ points should be "bid" in the whole set). There is obviously $\binom{n+m}{n}$ different ways to decide this attribution.
Assume now that exactly $k$ points among the highest $m$ ones $\left\{x_{n+1},\ldots,x_{n+m}\right\}$ are bid orders. Then necessarily there are exactly $n-k$ bid orders and $k$ ask orders among lowest $n$ points $\left\{x_1,\ldots,x_{n}\right\}$. As a consequence, $A(p) = k = B(p)$ for any $p$ in $(x_n,x_{n+1})$. The monotonicity of $A$ and $B$ ensures that we cannot have $A(p) = B(p)$ outside this interval.
Henceforth, the traded volume is $k$ if and only if there is $k$ bid orders among the highest $m$ orders, and $n-k$ among the lowest $n$ orders. There is obviously $\binom{m}{k}\binom{n}{n-k}$ ways to do this, hence the first result given at equation \eqref{eq:ConditionalTradedVolume}.
Finally, the same argument stating the monotonicity of $A$ and $B$ ensures that a necessary and sufficient condition for a price $p$ to be an exchange price maximizing the traded volume is to lie in the price interval $(x_n,x_{n+1})$, hence the second and third results.

For the sake of completeness, we provide the sketch of the full proof by calculus, leaving the details of the computation to the reader. Let us consider the call auction closed. Let $(X^A_i)_{i=1,\ldots,A(\infty)}$ be the (almost surely finite) sequence of prices of the submitted ask orders, and $(X^B_j)_{j=1,\ldots,B(-\infty)}$ be the (almost surely finite) sequence of prices of the submitted bid orders. Using standard notations for order statistics, let $X^A_{(i)}$ denote the $i$-th ask order when ordered increasingly in price : $X^A_{(1)}<\ldots<X^A_{(A(\infty))}$. Similarly, $X^B_{(j)}$ is the $j$-th bid order when ordered increasingly in price. We will also use the following conventions : $X^A_{(0)}=X^B_{(0)}=-\infty$ and $X^A_{(N)}=+\infty$ (resp. $X^B_{(N)}=+\infty$) if $N>A(\infty)$ (resp. $N>B(-\infty)$).
One can verify that:
\begin{align}
	& \mathbf P\left(V=k \vert A(\infty)=m,B(-\infty)=n\right)
	\nonumber \\
	= & \mathbf P\left(X^A_{(k)}<X^B_{(n-k+1)}, X^A_{(k+1)}>X^B_{(n-k)} \vert A(\infty)=m, B(-\infty)=n\right).
\end{align}
Now, using the independence of the bid and ask processes and knowing the joint distributions of any couple of order statistics \citep[see e.g.][for textbooks on order statistics]{David2003,Arnold2008}, we write the density of the quadruplet $(X^A_{(k)}, X^A_{(k+1)},X^B_{(n-k)},X^B_{(n-k+1)})$ conditionally to $A(\infty)=m,B(-\infty)=n$ as the function $g$ defined for any reals $a<b, c<d$ by:
\begin{align}
	g(a,b,c,d) = & \frac{m!}{(k-1)!(m-k-1)!} [F(a)]^{k-1} [1-F(b)]^{m-k-1} f(a) f(b)
	\nonumber \\
	& \times \frac{n!}{(k-1)!(n-k-1)!} [F(c)]^{n-k-1} [1-F(d)]^{k-1} f(c) f(d).
\end{align}
Integrating in $d$ then $c$ gives after some computation:
\begin{align}
	& \mathbf P\left(V=k \vert A(\infty)=m,B(-\infty)=n\right) \nonumber
	\\ = & \int_{\mathbb R} \int_{]a,+\infty[} \frac{m!}{(k-1)!(m-k-1)!} [F(a)]^{k-1} [1-F(b)]^{m-k-1} f(a) f(b) \nonumber
	\\ & \times \left[ F_{X_{(n-k+1)}}(b) - F_{X_{(n-k+1)}}(a) + \binom{n}{k} [F(b)]^{n-k} [1-F(b)]^{k} \right] 
	\, \textrm{d}b \,\textrm{d}a
\end{align}
where $F_{X_{(n-k+1)}}$ is the cumulative distribution function of the $(n-k+1)$-th order statistics of a sample of size $n$.
Recall that if $\beta(x,i,j) \triangleq \int_0^x u^{i-1}(1-u)^{j-1} \,du$ is the incomplete beta function, then $F_{X_{(n-k+1)}}(u)= \frac{n!}{(n-k)!(k-1)!}\beta(F(u),n-k+1,k)$.
Using the changes of variables $u=F(a)$ and $v=F(b)$ and some computations involving standard identities for the beta functions yields the result given at equation \eqref{eq:ConditionalTradedVolume}.
\end{proof}

Thanks to this simple but yet unnoted result (to our knowledge), we can proceed to the analytical exact and asymptotic solutions for the general problem.

\begin{rmk}
\label{rmk:SameBuySellPriceDistribution}
We underline that Lemma \ref{lem:ConditionalDistributions} does not stand if the distributions of the prices of ask orders $F_A$ and bid orders $F_B$ are assumed to be different. The first proof given above uses the fact that all "bid" and "ask" tag attributions to the submitted orders are equiprobable when $F_A=F_B=F$. If $F_A\neq F_B$, then this is not the case and these probabilities are price-dependent: the probability that an order submitted at a price lower or equal than $p$ is an ask order is $\frac{\lambda_A F_A(p)}{\lambda_A F_A(p)+\lambda_B F_B(p)}$. Therefore, any reasoning mimicking the first proof would require the evaluation of the traded volume conditionally on the placement of the $n$-th order, which does not appear to be easily tractable.
Trying to go as a workaround through the proof by direct calculus using $F_A$ and $F_B$ (absolutely continuous with densities $f_A$ and $f_B$), one would reach integrals of the general form :
$$\int_{]a,+\infty[} f_A(u) (1-F_A(u))^i \beta(F_B(u),j,k) \,du,
$$ 
which are tractable when $F_A = F_B$, but do not appear to be obviously so when $F_A\neq F_B$.
\end{rmk}

\section{Exact and asymptotic distributions of the traded volume}
\label{sec:TradedVolumeDistribution}

The Poisson assumption for the bid and ask processes immediately yields the general form of the unconditional distribution of the traded volume. With a little rewriting one straightforwardly gets the general result.

\begin{proposition} 
\label{prop:UnconditionalTradedVolume}
In the general call auction model with orders submission rate $\lambda$, market imbalance $\alpha$ and auction length $T$, the distribution of the traded volume is for any $k\in\mathbb N$:
\begin{equation}
	\mathbf P(V=k) 
	= e^{-\lambda T} \frac{\left(\alpha(1-\alpha)\lambda^2 T^2\right)^k}{\left(k!\right)^2}
	\,\,\sum_{i=0}^{+\infty}\sum_{j=0}^{+\infty} \frac{(\lambda T)^{i+j}}{\binom{i+j+2k}{i+k}}
	\frac{\alpha^i}{i!} \frac{(1-\alpha)^j}{j!}
\end{equation}
Note that this result does not depend on the orders' price distribution $F$.
\end{proposition}

This form has symmetry properties (with respect to the processes of buy and sell orders, expressed through the market imbalance $\alpha\in[0,1]$), and it is therefore convenient for symbolic computation. However, any of the two sums can be nicely expressed using special functions, which can be of interest in terms of computational implementation. For example,
\begin{equation}
	\mathbf P(V=k) 
	= e^{-\lambda T} \left(\alpha(1-\alpha)\lambda^2 T^2\right)^k
	\,\, \sum_{i=0}^{+\infty} \binom{k+i}{k} \frac{\left(\alpha \lambda T\right)^i}{(i+2k)!}
	\,\, {}_1F_1\left(k+1, i+2k+1,(1-\alpha)\lambda T\right),
\end{equation}
where ${}_1F_1$ is the confluent hypergeometric function \citep[see e.g.][]{Seaborn1991}. 

The result of proposition \ref{prop:UnconditionalTradedVolume} generalizes to a general market imbalance the distribution found in \citet{Mendelson1982} in the symmetric case with no market imbalance. By setting $\alpha=\frac{1}{2}$ in the result of proposition \ref{prop:UnconditionalTradedVolume}, and using some combinatorial rewriting including variants of the Vandermonde identity \citep{Gould1956}, we obtain :
\begin{equation}
	\mathbf P(V=k) = e^{-\frac{\lambda T}{2}} \left(\frac{\lambda T}{2}\right)^{2k} \frac{1}{(2k)!} \left(1 + \frac{\lambda T}{2(2k+1)}\right)
\end{equation}
which was found by \citet[][equation 3.3]{Mendelson1982}.
Furthermore, by setting $k=0$, proposition \ref{prop:UnconditionalTradedVolume} gives the probability that no trade will occur at time $T$ during the clearing mechanism, which is an important quantity in very illiquid markets, as well as the influence of the market imbalance $\alpha$ on this quantity. Straightforward identities give for $k=0$ and $\alpha\neq\frac{1}{2}$:
\begin{equation}
	\mathbf P(V=0) = e^{-\alpha\lambda T} + \frac{\alpha}{1-2\alpha}\left( e^{-\alpha\lambda T} - e^{-(1-\alpha)\lambda T}\right).
\end{equation}
Once again, in the symmetric case $\alpha=\frac{1}{2}$, one retrieves by computing the limit of the above expression:
\begin{equation}
\mathbf P(V=0) = e^{-\frac{\lambda T}{2}} \left(1 + \frac{\lambda T}{2}\right),
\end{equation}
which was obtained by \citet[unnumbered equation p.1512]{Mendelson1982}.

Figure \ref{fig:TradedVolumeDistribution} plots several examples of the traded volume distribution in the case of an illiquid market ($\lambda=10$) and in the case of a more liquid market ($\lambda=100$).
In the example cases of the illiquid market, the probability that no trade happens varies from $4\%$ to $33\%$ depending on the market imbalance $\alpha$. In the example cases of the liquid market, this probability is less than $10^{-5}$ even when the market is not balanced.
\begin{figure}
\centering
\begin{tabular}{cc}
	\includegraphics[scale=0.5]{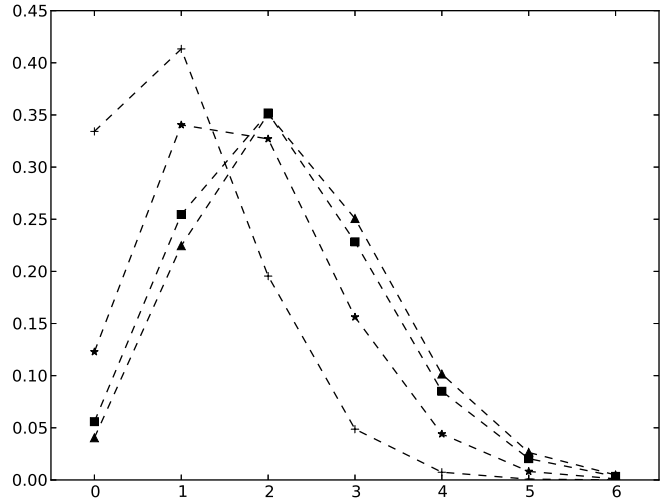}
	\\
	\includegraphics[scale=0.5]{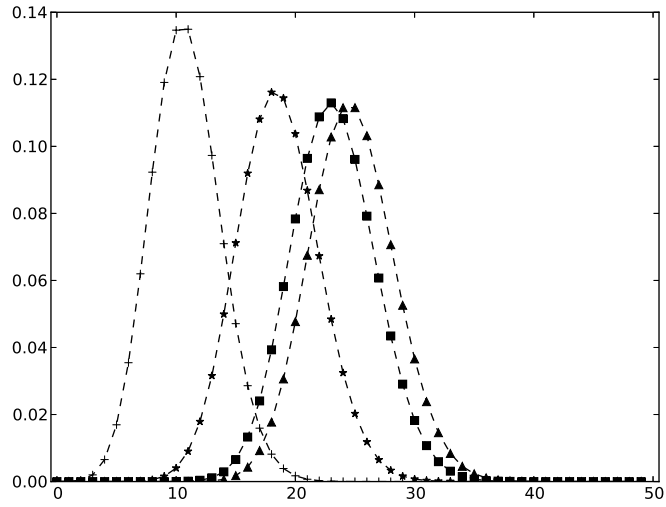}	
\end{tabular}
\caption{Traded volume distribution (volume on the $x$-axis, mass probabilities on the $y$-axis) for several values of market imbalance : $\alpha=0.125$ (plus), $\alpha=0.25$ (star), $\alpha=0.375$ (square), $\alpha=0.5$ (triangle) and several types of market liquidity: $\lambda=10$ (top) and $\lambda=100$ (bottom).}
\label{fig:TradedVolumeDistribution}
\end{figure}
As expected, the distribution is shifted to the left as the market imbalance increases, i.e. $\left|\alpha-\frac{1}{2}\right|$ increases.

As hinted by figure \ref{fig:TradedVolumeDistribution}, the traded volume is asymptotically normal as the market liquidity increases. The mean is $\lambda T\alpha(1-\alpha)$, which might have been guessed by taking the equivalent of the mean of the hypergeometric distribution of lemma \ref{lem:ConditionalDistributions} as $N\to+\infty$, but the variance of the traded volume, however, is not an intuitive function of the market imbalance. The rigorous result is the following.
\begin{proposition}
\label{prop:AsymptoticVolume}
In the general call auction model with orders submission rate $\lambda$, market imbalance $\alpha$ and auction length $T$, in the very liquid case where $\lambda T\to+\infty$,
the distribution of the traded volume is asymptotically normal with mean $\lambda T \alpha (1-\alpha)$ and standard deviation $\sqrt{\lambda T \alpha(1-\alpha)(1-2\alpha(1-\alpha))}$.
Note that this result does not depend on the orders' price distribution $F$.
\end{proposition}
\begin{proof}
The problem is to determine the weak limit of the hypergeometric distribution of lemma \ref{lem:ConditionalDistributions} when the number of ask orders is Poisson with parameter $\alpha\lambda T$ and the number of bid orders is Poisson with parameter $(1-\alpha)\lambda T$, or, equivalently, when the total number of orders $N$ is Poisson with parameter $\lambda T$ and the conditional number of bid orders $n_N$ is binomial with parameters $N$ and $(1-\alpha)$.
The proof is based on two successive applications of limit theorems for sequences of random variables with random indices, the first one on $n_N$ when $N$ is assumed deterministic, and the second one on $N$ finally assumed random.

Let $V_{N,n}$ a random variable with hypergeometric distribution $\mathbf P(V_{N,n}=k)=\frac{\binom{n}{k}\binom{N-n}{n-k}}{\binom{N}{N-n}}$, $0\leq k\leq n\leq N$.
Assume $(n_N)_{N}$ is a deterministic sequence of natural integers such that $\frac{n_N}{N}\underset{N\to+\infty}\longrightarrow 1-\alpha$. Then $V_{N,n_N}$ is weakly convergent to a Gaussian distribution \citep[see e.g.][section VII.7, Problem 10]{Feller1968}:
\begin{equation}
	\frac{V_{N,n_N}-N\alpha(1-\alpha)}{\sqrt{N}\alpha(1-\alpha)} 
	\overset{d}{\underset{N\to+\infty}\longrightarrow} \mathcal N(0,1) .
\end{equation}
Assume now that $n_N$ is a random variable with binomial distribution $\mathcal B(N,1-\alpha)$, and let
$V_N=V_{N,n_N}$. Let also $l_N(s)$ denote the infimum of the $s$-quantile of $n_N$.
Standard convergence of the binomial distribution is written:
\begin{equation}
	\frac{n_N-N(1-\alpha)}{\sqrt{N\alpha(1-\alpha)}}
	\overset{d}{\underset{N\to+\infty}\longrightarrow} \mathcal N(0,1),
\end{equation}
hence we may write: $l_N(s) = N(1-\alpha)+ \Phi^{-1}(s)\sqrt{N\alpha(1-\alpha)}+o(\sqrt{N})$.
Let us set $1-\alpha_N(s) \triangleq \frac{l_N(s)}{N}$, $a_N(s) \triangleq N \alpha_N(s) ( 1-\alpha_N(s) )$ and $b_N(s) \triangleq \sqrt{N} \alpha_N(s)( 1-\alpha_N(s) )$.
Straightforward computations yield:
\begin{equation}
	a_N(s) = N\alpha(1-\alpha) + (2\alpha-1)\sqrt{N\alpha(1-\alpha)} \Phi^{-1}(s) + o(\sqrt{N}),	
\end{equation}
and
\begin{equation}
	b_N(s) = \sqrt{N}\alpha(1-\alpha) + o(\sqrt{N}).
\end{equation}
Therefore, $\frac{a_N(s)-N\alpha(1-\alpha)}{\sqrt{N}\alpha(1-\alpha)} \underset{N\to+\infty}\longrightarrow \Phi^{-1}(s) \frac{2\alpha-1}{\sqrt{\alpha(1-\alpha)}}$ and $\frac{b_N(s)}{\sqrt{N}\alpha(1-\alpha)} \underset{N\to+\infty}\longrightarrow 1$. Hence, applying Theorem 4 of \citet{Korolev1993} on the weak convergence on sequences of random variables with random indices, we obtain that $\frac{V_N-N\alpha(1-\alpha)}{\sqrt{N}\alpha(1-\alpha)}$ converges weakly to the sum of two independant zero-mean normal distributions, one with standard deviation $1$ and one with standard deviation $\frac{\left|2\alpha-1\right|}{\sqrt{\alpha(1-\alpha)}}$. By summing the variances and rescaling, we obtain:
\begin{equation}
	\label{eq:ConvergenceVFixedNRandomn}
	\frac{V_N-N\alpha(1-\alpha)}{\sqrt{N\alpha(1-\alpha)(1-3\alpha(1-\alpha))}} 
	\overset{d}{\underset{N\to+\infty}\longrightarrow} \mathcal N(0,1).
\end{equation}
Finally, assume now that $N_k$ is a random variable with Poisson distribution with parameter $k=\lfloor\lambda T\rfloor$ and let $l_k(s)$ denote the infimum of the $s$-quantile of $N_k$.
Standard convergence of the Poisson distribution is written:
\begin{equation}
	\frac{N_k-k}{\sqrt{k}}
	\overset{d}{\underset{k\to+\infty}\longrightarrow} \mathcal N(0,1),
\end{equation}
hence we may write: $l_k(s) = k + \Phi^{-1}(s)\sqrt{k} + o(\sqrt{k})$.
Thanks to equation \eqref{eq:ConvergenceVFixedNRandomn} we have:
\begin{equation}
	\frac{V_{l_k(s)}-l_k(s)\alpha(1-\alpha)}{\sqrt{l_k(s)\alpha(1-\alpha)(1-3\alpha(1-\alpha))}} 
	\overset{d}{\underset{k\to+\infty}\longrightarrow} \mathcal N(0,1) .
\end{equation}
Let 
\begin{equation}
	a_k(s) \triangleq l_k(s)\alpha(1-\alpha) 
	= k\alpha(1-\alpha) + \sqrt{k}\alpha(1-\alpha) \Phi^{-1}(s) + o(\sqrt{k}),	
\end{equation}
and
\begin{equation}
	b_k(s) \triangleq \sqrt{l_k(s)\alpha(1-\alpha)(1-3\alpha(1-\alpha))}
	= \sqrt{k\alpha(1-\alpha)(1-3\alpha(1-\alpha))} + o(\sqrt{k}).
\end{equation}
This yields 
\begin{equation}
	\frac{a_k(s)-k\alpha(1-\alpha)}{\sqrt{k\alpha(1-\alpha)(1-3\alpha(1-\alpha))}} \underset{k\to+\infty}\longrightarrow \Phi^{-1}(s) \sqrt{\frac{\alpha(1-\alpha)}{1-3\alpha(1-\alpha))}}
\end{equation}
and
\begin{equation}
\frac{b_k(s)}{\sqrt{k\alpha(1-\alpha)(1-3\alpha(1-\alpha))}} \underset{k\to+\infty}\longrightarrow 1.
\end{equation}
Hence, applying once again Theorem 4 of \citet{Korolev1993}, we finally obtain that:
\begin{equation}
	\frac{V_k-k\alpha(1-\alpha)}{\sqrt{k\alpha(1-\alpha)(1-3\alpha(1-\alpha))}} 
	\overset{d}{\underset{k\to+\infty}\longrightarrow}
	\mathcal N\left(0,\sqrt{\frac{1-2\alpha(1-\alpha)}{1-3\alpha(1-\alpha)}}\right),
\end{equation}
which, with one last rescaling, completes the proof.
\end{proof}

The first moment of the distribution (but not the distribution) was approximately derived by \citet[equation 3.1]{Mendelson1982} by resorting to asymptotic results from renewal theory. Using our notations, it was found there that if the market is liquid enough (i.e. $\lambda \gg 1$), then the expected traded volume is approximately the inverse of the expectation of the sum of two independent exponential variables with parameters $\alpha\lambda T$ and $(1-\alpha)\lambda T$, which gives $\alpha(1-\alpha)\lambda T$ and is indeed the average traded volume obtained in proposition \ref{prop:AsymptoticVolume} when $\lambda$ is large.

\section{Exact and asymptotic distributions of the lower and upper prices}
\label{sec:PricesDistributions}

We now extend lemma \ref{lem:ConditionalDistributions} to obtain the unconditional distributions of the lower and upper bound of the possible clearing prices interval.
Following lemma \ref{lem:ConditionalDistributions}, the density of the lower bound of the interval of potential clearing prices conditionally to the set $\{A(\infty)=m,B(-\infty)=n\}$ is:
\begin{equation}
	f_{L \vert A(\infty)=m,B(-\infty)=n} (x) = \frac{(n+m)!}{(n-1)!m!} \left[F(x)\right]^{n-1} \left[1-F(x)\right]^{m} f(x),
\end{equation}
which gives the following result.

\begin{proposition}
In the general call auction model with orders submission rate $\lambda$, market imbalance $\alpha$, auction length $T$ and price distribution $F$ (absolutely continuous with density $f$), the distribution of the lower bound $L$ of the possible clearing prices admits the probability density function $f_L$ defined for any $x\in\mathbb R$ as:
\begin{align}
	f_L(x) = & \left(e^{\alpha\lambda T}-1\right)^{-1} \left(e^{(1-\alpha)\lambda T}-1\right)^{-1} f(x)
	\nonumber \\
	& \sum_{n=1}^{+\infty} \frac{\left[(1-\alpha)\lambda T\right]^n}{n!} \frac{\left[F(x)\right]^{n-1}}{(n-1)!}
	\sum_{m=1}^{+\infty} (n+m)! \frac{\left[\alpha \lambda T\right]^m}{m!} \frac{\left[1-F(x)\right]^{m}}{m!}	
\end{align}
Symmetrically, the probability density function $f_U$ of the uper bound of the clearing prices range is obtained by substituting $\alpha$ with $1-\alpha$ and $F(x)$ with $1-F(x)$ in the above formula.
\end{proposition}

Here again, we have highlighted the symmetric result in the proposition, but carrying out some computations may lead to express these densities with a single series, such as:
\begin{align}
	f_L(x) = \frac{(1-\alpha)\lambda T f(x) e^{(1-\alpha)\lambda T F(x)}}{\left(e^{\alpha\lambda T}-1\right)\left(e^{(1-\alpha)\lambda T}-1\right)}
	& \Big[ -1 + e^{-(1-\alpha)\lambda T F(x)}
	\\ 
	& \times \sum_{n=0}^{+\infty} \frac{\left((1-\alpha)\lambda T F(x)\right)^n}{n!}
	{}_1F_1\left( n+2, 1, \alpha \lambda T (1-F(x)) \right)
	\Big]. \nonumber 
\end{align}

Figure \ref{fig:LowerUpperPricesDistribution} plots several examples of this distribution, showing the influence of the market imbalance $\alpha$ in the liquid and non-liquid market cases.
\begin{figure}
\centering
\begin{tabular}{cc}
	\includegraphics[scale=0.5]{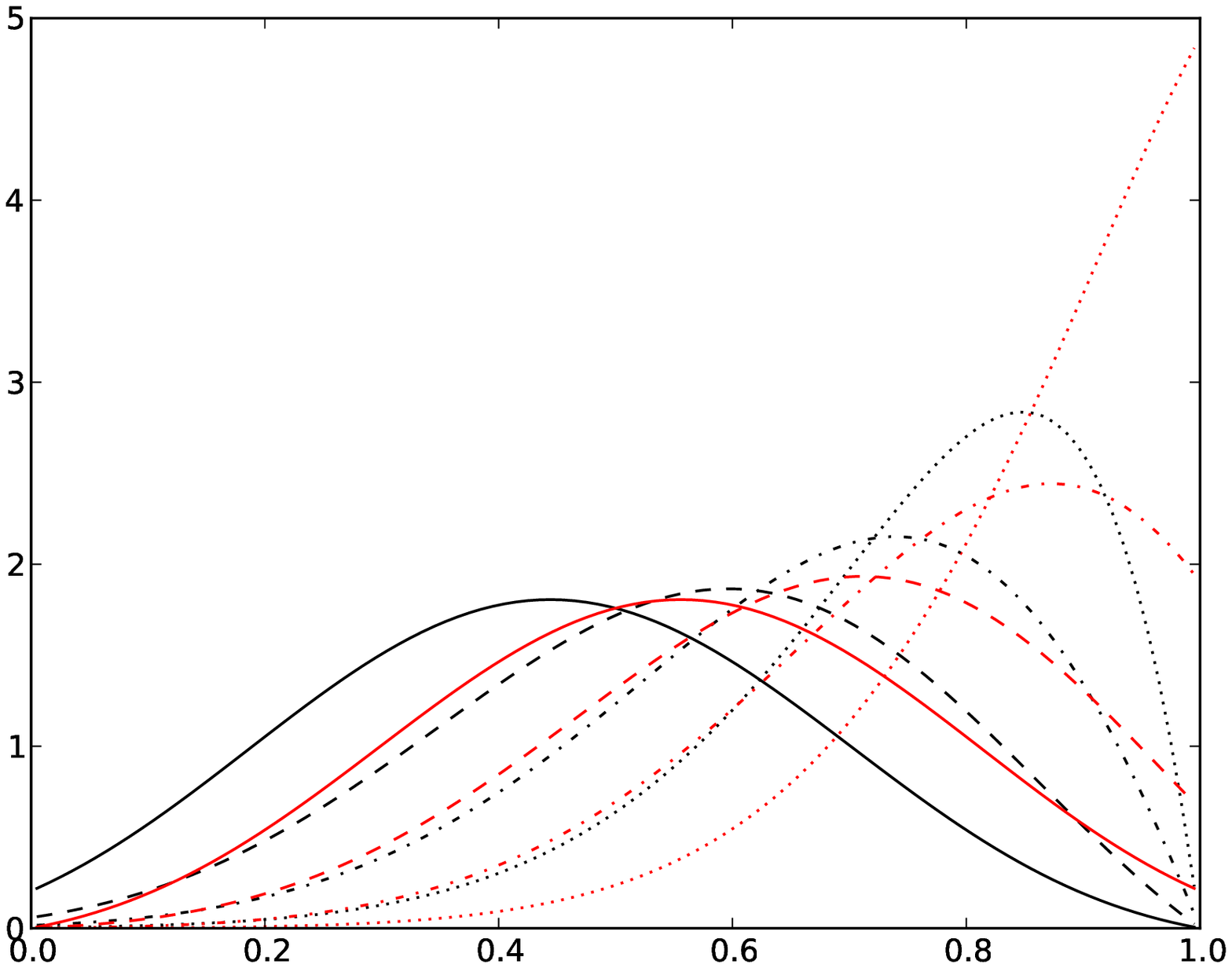}
	\\
	\includegraphics[scale=0.5]{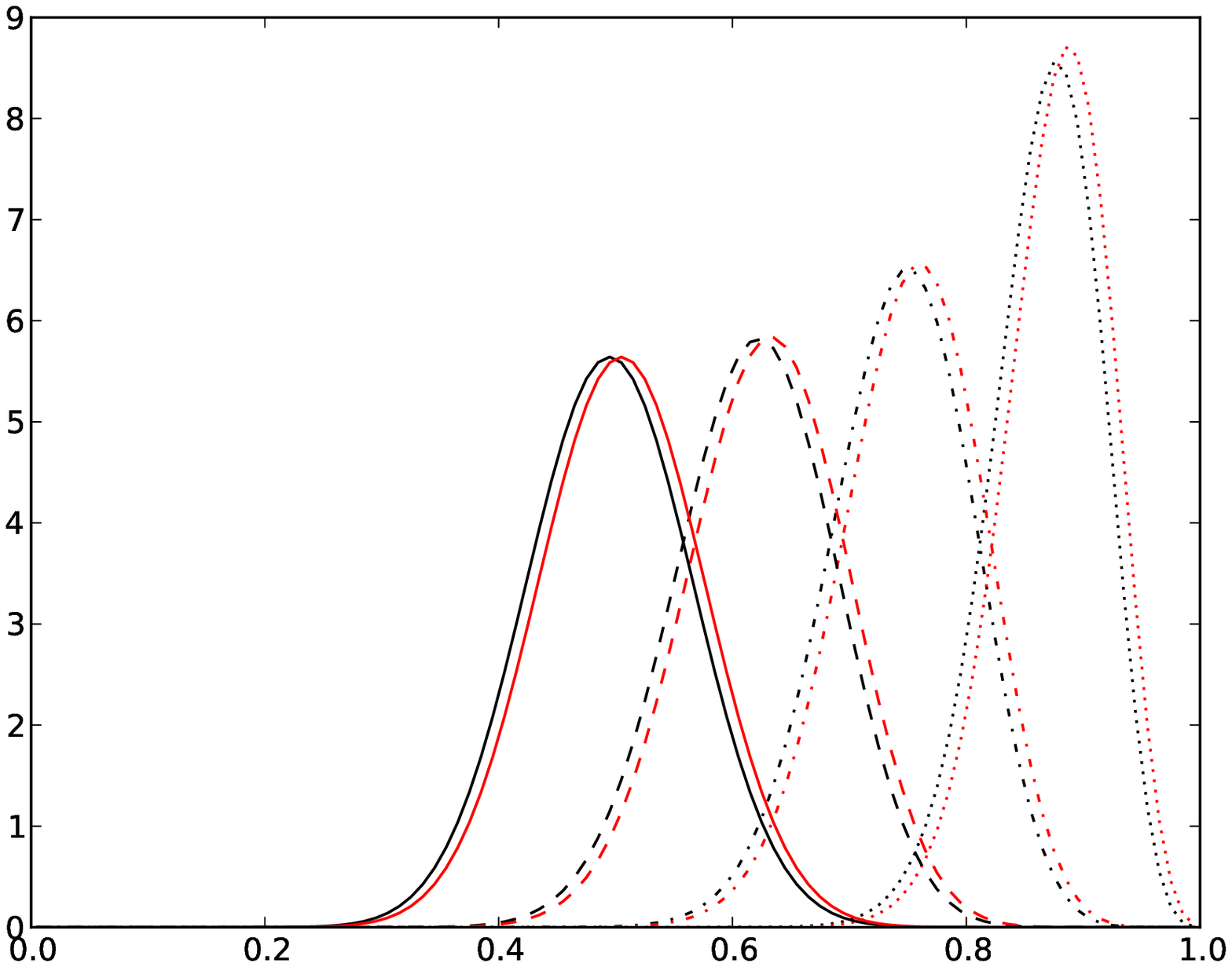}	
\end{tabular}
\caption{Lower (in black) and upper (in red) prices distributions (price on the $x$-axis, probability density functions on the $y$-axis) for several values of market imbalance $\alpha=0.125$ (dotted), $\alpha=0.25$ (dash-dotted), $\alpha=0.375$ (dashed), $\alpha=0.5$ (full) and several types of market liquidity: $\lambda=10$ (top) and $\lambda=100$ (bottom). Price distribution is uniform on $(0,1)$ in these examples.}
\label{fig:LowerUpperPricesDistribution}
\end{figure}
One guesses on these graphs that both the lower and upper prices distributions converges to the same normal distribution as the market liquidity increases.
This result is formally stated and proven in the next proposition.

\begin{proposition}
\label{prop:AsymptoticPrices}
In the general call auction model with orders submission rate $\lambda$, market imbalance $\alpha$, auction length $T$ and price distribution $F$ (absolutely continuous with density $f$),
in the very liquid case where $\lambda T\to+\infty$,
the distribution of the lower and upper bounds of the possible clearing prices is asymptotically normal
with mean $F^{-1}(1-\alpha)$
and standard deviation $\frac{1}{f(F^{-1}(1-\alpha))}\sqrt{\frac{2\alpha(1-\alpha)}{\lambda T}}$.
\end{proposition}
\begin{proof}
Lemma \ref{lem:ConditionalDistributions} proves that conditionally on counting a total of $N$ orders including $n$ bid orders, the lower price $L_{N,n}$ is distributed as the $n$-th statistics of a random sample with size $N$ distributed according to $F$.
Assume temporarily that $N$ is fixed and $(n_N)_N$ is an sequence of integers such that $\frac{n_N}{N}\underset{N\to+\infty}\longrightarrow 1-\alpha$. Then the weak convergence of central order statistics \citep[see e.g.][Theorem 10.3]{David2003} states that:
\begin{equation}
	\label{eq:WeakConvergencenthOrderStatistics}
	\sqrt{N} \left( L_{N,n_N}-F^{-1}(1-\alpha) \right)
	\overset{d}{\underset{N\to+\infty}\longrightarrow}
	\mathcal N\left(0, \frac{\sqrt{\alpha(1-\alpha)}}{f(F^{-1}(1-\alpha))} \right).
\end{equation}
We now proceed by mimicking the proof of proposition \ref{prop:AsymptoticVolume}, omitting details for the sake of brevity.
Assume now that $n_N$ is distributed according to a binomial distribution $\mathcal B(N,1-\alpha)$. If $l_N(s)$ is the infimum of the $s$-quantile of this distribution and $1-\alpha_N(s)\triangleq\frac{l_N(s)}{N}$, then by equation \eqref{eq:WeakConvergencenthOrderStatistics}:
\begin{equation}
	\frac{L_{N,l_N(s)}-F^{-1}(1-\alpha)}{\frac{1}{f(F^{-1}(1-\alpha)}\sqrt{\frac{\alpha(1-\alpha)}{N}}}
	\overset{d}{\underset{N\to+\infty}\longrightarrow}
	\mathcal N\left(0,1\right),
\end{equation}
and the weak convergence of the binomial distribution leads to a simple expansion for the inverse of the absolutely continuous distribution $F$, which is written:
\begin{equation}
	\label{eq:Fm1expansion}
	F^{-1}(1-\alpha_N(s)) = F^{-1}(1-\alpha) 
	+ \frac{1}{f(F^{-1}(1-\alpha))}\sqrt{\frac{\alpha(1-\alpha)}{N}} \Phi^{-1}(s)
	+ o\left( N^{-1/2}\right).
\end{equation}
At this point, following \citet[Theorem 4]{Korolev1993}, the second term of equation \eqref{eq:Fm1expansion} is the source of a new standard gaussian term in the convergence, and we obtain that: 
\begin{equation}
	\sqrt{N} \left( L_{N,n_N}-F^{-1}(1-\alpha) \right) 
	\overset{d}{\underset{N\to+\infty}\longrightarrow}
	\mathcal N\left(0, \sqrt{2} \frac{\sqrt{\alpha(1-\alpha)}}{f(F^{-1}(1-\alpha)}\right).
\end{equation}
Finally, lifting the deterministic assumption on $N$ (i.e. reassuming that $N=\lfloor\lambda T\rfloor$ is Poisson) and re-applying the limit theorem for random indices leads to the weak convergence of the result of the proposition.

Note that the proof for the convergence of the upper price distribution is identical, since conditionally on counting a total of $N$ orders including $n$ bid orders, the upper price $U_{N,n}$ is distributed as the $n+1$-th statistics of a random sample with size $N$ distributed according to $F$.
\end{proof}

Figure \ref{fig:PriceConvergence} illustrates the convergence of proposition \ref{prop:AsymptoticPrices} as the liquidity increases, using empirical simulations of the model.
\begin{figure}
\centering
\includegraphics[scale=0.5]{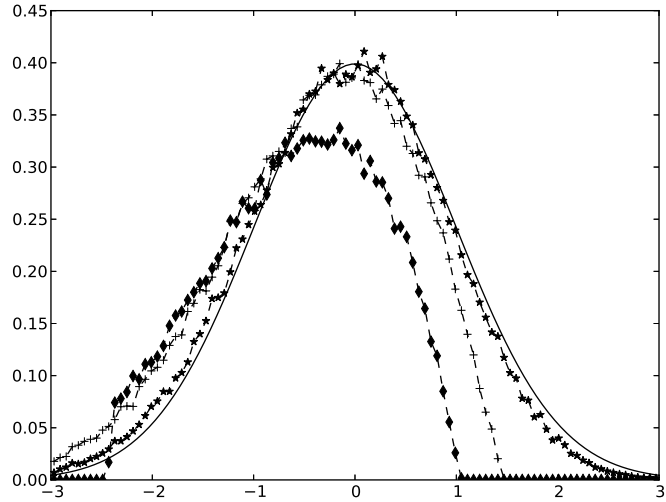}
\caption{Empirical densities of the lower prices centered and rescaled by the location and scale parameters of proposition \ref{prop:AsymptoticPrices}, as the liquidity increases : $\lambda=5$ (diamond), $\lambda=10$ (plus), $\lambda=100$ (stars). Full line is the standard gaussian. In these examples, empirical prices are obtained by simulation with $T=1$, $\alpha=0.3$ and prices uniformly distributed on $(0,1)$. The empirical distribution is computed with $10^5$ simulations.}
\label{fig:PriceConvergence}
\end{figure}

\section{Exact and asymptotic distributions of the range of clearing prices}
\label{sec:PriceRangeDistribution}

We now focus on the range of potential clearing prices, which is a proxy to the post-clearing spread, and is therefore, along with the traded volume, a fundamental variable from a liquidity point of view.

Thanks to the results of lemma \ref{lem:ConditionalDistributions}, we know that the clearing range $R$ is, conditionnally on the set $\left\{A(\infty)=m,B(-\infty)=n\right\}$, the range between the $n$-th and the $(n+1)$-th order statistics of an i.i.d. sample of size $n+m$ with distribution $F$. It conditional density $f_{R\vert A(\infty)=m,B(-\infty)=n}$ is therefore written :
\begin{equation}
	f_{R\vert A(\infty)=m,B(-\infty)=n}(\delta) = \frac{(n+m)!}{(n-1)!(m-1)!}\int_{\mathbb R} 
	\left[F(x)\right]^{n-1} f(x) f(x+\delta) \left[1-F(x+\delta)\right]^{m-1} \,dx.
\end{equation}
In the special case where $F$ is the uniform distribution on $(0,1)$, this conditional density can be explicitly computed as: 
\begin{equation}
	f_{R\vert A(\infty)=m,B(-\infty)=n}(\delta) = (n+m)(1-\delta)^{n+m-1},
\end{equation}
which in turn yields the result of proposition \ref{prop:UnconditionalPriceRange-UnifPrice}.

\begin{proposition}
\label{prop:UnconditionalPriceRange-UnifPrice}
In the call auction model with orders submission rate $\lambda$, market imbalance $\alpha$, auction length $T$ \emph{and uniform orders' price distribution on $(0,1)$}, the distribution of the clearing price range admits the probability density function $f_R$ defined for any $\delta\in(0,1)$ as:
\begin{equation}
	f_R(\delta) = \frac{\lambda T e^{-\lambda T\delta}}{\left(1-e^{-\alpha\lambda T}\right)\left(1-e^{-(1-\alpha)\lambda T}\right)} 
	\left[
	1 - (1-\alpha) e^{-\alpha \lambda T(1-\delta)} - \alpha e^{-(1-\alpha)\lambda T (1-\delta) }
	\right].
\end{equation}
\end{proposition}

Thus it appears that in the case of uniform prices, the distribution of the clearing price range is an exponential distribution with parameter $\lambda T$ modified by exponential terms functions of the market imbalance.
The general case does not exhibit such a simple formula for the exact distribution. However, the next proposition proves that, even in the general case, the clearing price range is asymptotically exponential.
\begin{proposition}
\label{prop:AsymptoticPriceRange}
In the call auction model with orders submission rate $\lambda$, market imbalance $\alpha$, auction length $T$ and price distribution $F$ (absolutely continuous with continuous density $f$), 
in the very liquid case where $\lambda T\to+\infty$,
the distribution of the clearing price range is asymptotically exponential with parameter (inverse of the mean) $\lambda T f(F^{-1}(1-\alpha))$.
\end{proposition}
\begin{proof}
Conditionally on counting a total of $N$ orders including $n$ bid orders, the clearing price range $R_{N,n}$ is distributed as the range between the $n$-th and the $n+1$-th statistics of a random sample with size $N$ distributed according to $F$.
If $(n_N)_N$ is a sequence of integers such that $\frac{n_N}{N}\underset{N\to\infty}\longrightarrow1-\alpha$,
then the distribution of $N R_{N,n_N}$ is known to be weakly convergent to an exponential distribution with parameter $f(F^{-1}(1-\alpha))$ \citep[see e.g.][p.328, refering to \citep{Pyke1965}]{David2003}.
The remaining part of the proof mimicks the previous proofs, and as a consequence it is omitted for brevity. It turns out that when $N$ is Poisson with parameter $k$ and $n_N$ is binomial with parameters $N$ and $1-\alpha$, $N R_{N,n_N}$ has the same weak limit as in the deterministic case, hence the result.
\end{proof}

Figure \ref{fig:PriceRangeConvergence} illustrates the convergence of proposition \ref{prop:AsymptoticPriceRange} as the liquidity increases, using empirical simulations of the model.
\begin{figure}
\centering
\includegraphics[scale=0.5]{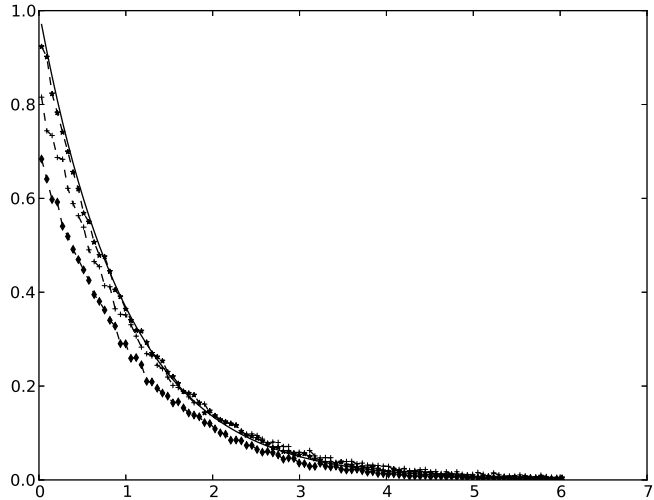}
\caption{Empirical densities of the clearing price range scaled by the parameter of proposition \ref{prop:AsymptoticPriceRange} as the liquidity increases : $\lambda=5$ (diamond), $\lambda=10$ (+), $\lambda=50$ (stars). In these examples, empirical price ranges are obtained by simulation with $T=1$, $\alpha=0.3$ and orders' prices normally distributed with mean $0$ and variance $1$.}
\label{fig:PriceRangeConvergence}
\end{figure}

As for the scaling of proposition \ref{prop:AsymptoticPriceRange}, the fact that the clearing price range scales as the inverse of $\lambda T$ as the liquidity increases is in agreement with \citet[][equation 4.2]{Mendelson1982}. Note also that this in agreement with the scaling of the spread suggested using dimensional analysis by \citet{Smith2003}, although in a very different trading environment (a continuous double auction).

Moreover, as for the limit distribution of proposition \ref{prop:AsymptoticPriceRange}, the basic Poisson model for the call auction thus suggests that the exponential distribution is an acceptable zero-intelligence approximation for the post-clearing spread distribution.
Recall that previous empirical results on the distribution of the bid-ask spread during a continuous double auction have highlighted the fat tails of these distribution \citep[see e.g.][on U.S. and Chinese stocks respectively]{Plerou2005,Gu2007}. However, the spread as any other financial quantity exhibits intraday variations \citep[see e.g.][among early works]{Chan1995}. In particular, the spread observed right after a call auction is significantly larger than the one observed during a continuous double auction \citep[see][figure 2]{Gu2007} and may therefore exhibit different statistical properties that, to our knowledge, are yet to be studied.

In a first step in this direction, we conclude this section by a brief empirical illustration of the exponential distribution as an approximation of the distribution of the post-clearing spread. 
For a stock that is part of the CAC 40 index, each trading day at the Paris stock Exchange starts with an opening call auction similar to the one studied in this paper \citep{Euronext2013,Euronext2014}. At 9:00 a.m., the auction is cleared, and there results a post-clearing spread (and book) which is the initial state for the continuous double auction that starts at this time.
We have randomly selected $12$ stocks from the CAC 40 index : Air Liquide (AIRP.PA, chemicals), Alstom
(ALSO.PA, industrial machinery), Axa (AXAF.PA, insurance), BNP Paribas (BNPP.PA, banking), Bouygues (BOUY.PA, construction, telecom and media), Carrefour (CARR.PA, retail/wholesale), Danone (DANO.PA, food products), Michelin (MICP.PA, tires), Renault (RENA.PA, automobiles), Sanofi (SASY.PA, pharmaceuticals), Vinci (SGEF.PA, construction), Total (TOTF.PA, oil and gas).
For each of these $12$ stocks, we use an extraction from the Thomson Reuters Tick History database to get the opening bid and ask quotes from March 3rd, 2011 to June 28th, 2013. After data cleaning (missing days), we obtain a sample of post-clearing spreads of size 592 to 595, depending on the stock.
Figure \ref{fig:SpreadExpFit-DistributionFunction} plots, for each of these $12$ stocks, the distribution function of the post-clearing spread, and the exponential distribution function fitted on the data using maximum-likelihood estimation.
To get an idea of the tail of the distribution (although remember we only have roughly 600 points), Figure \ref{fig:SpreadExpFit-SurvivalFunction} plots the logarithm of the empirical survival distribution function of the spread, and its fitted exponential counterpart (a straight line in the semilog scale used).

\begin{figure}
\centering
\begin{tabular}{ccc}
	\includegraphics[width=0.3\textwidth]{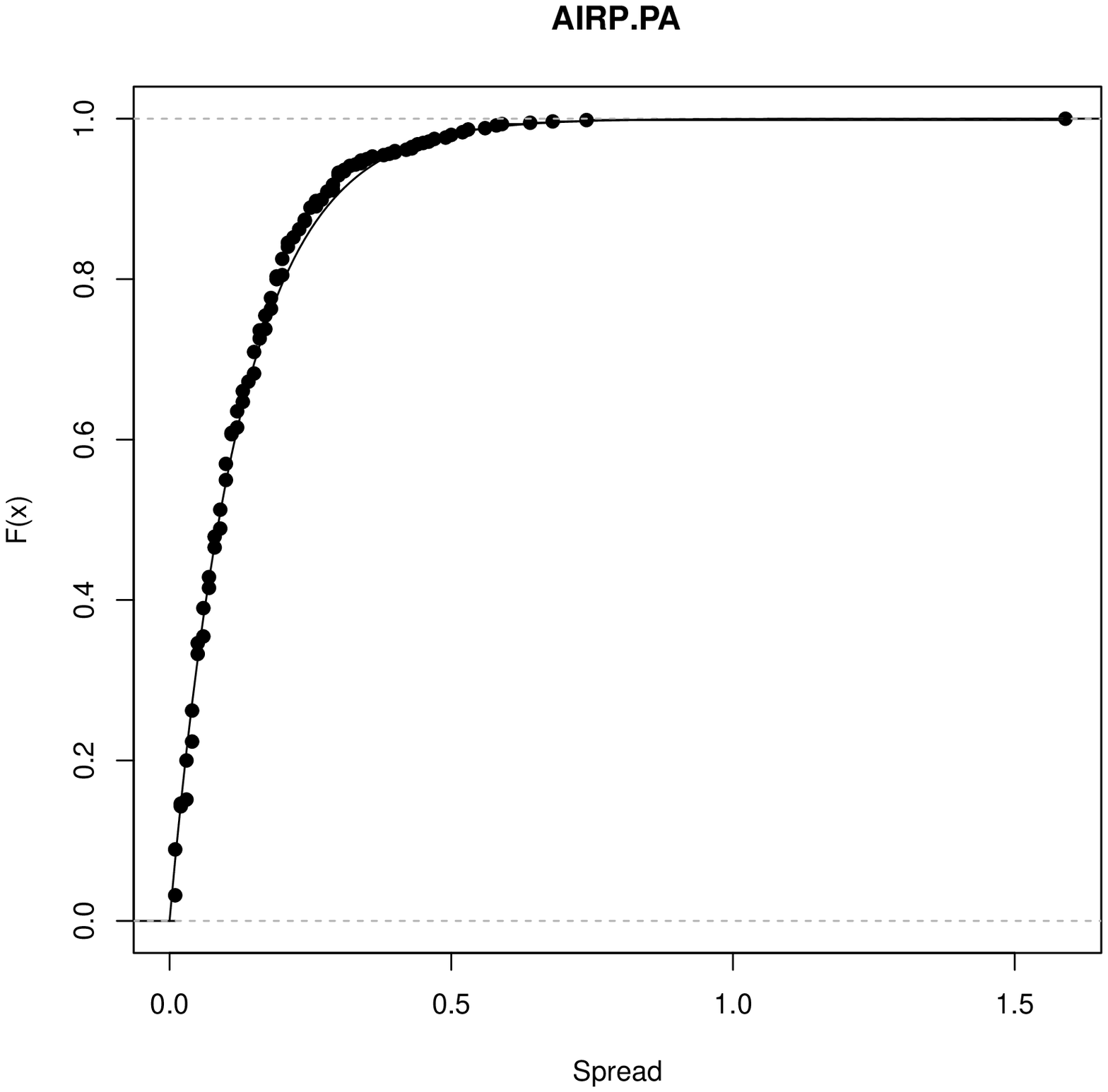}
	& \includegraphics[width=0.3\textwidth]{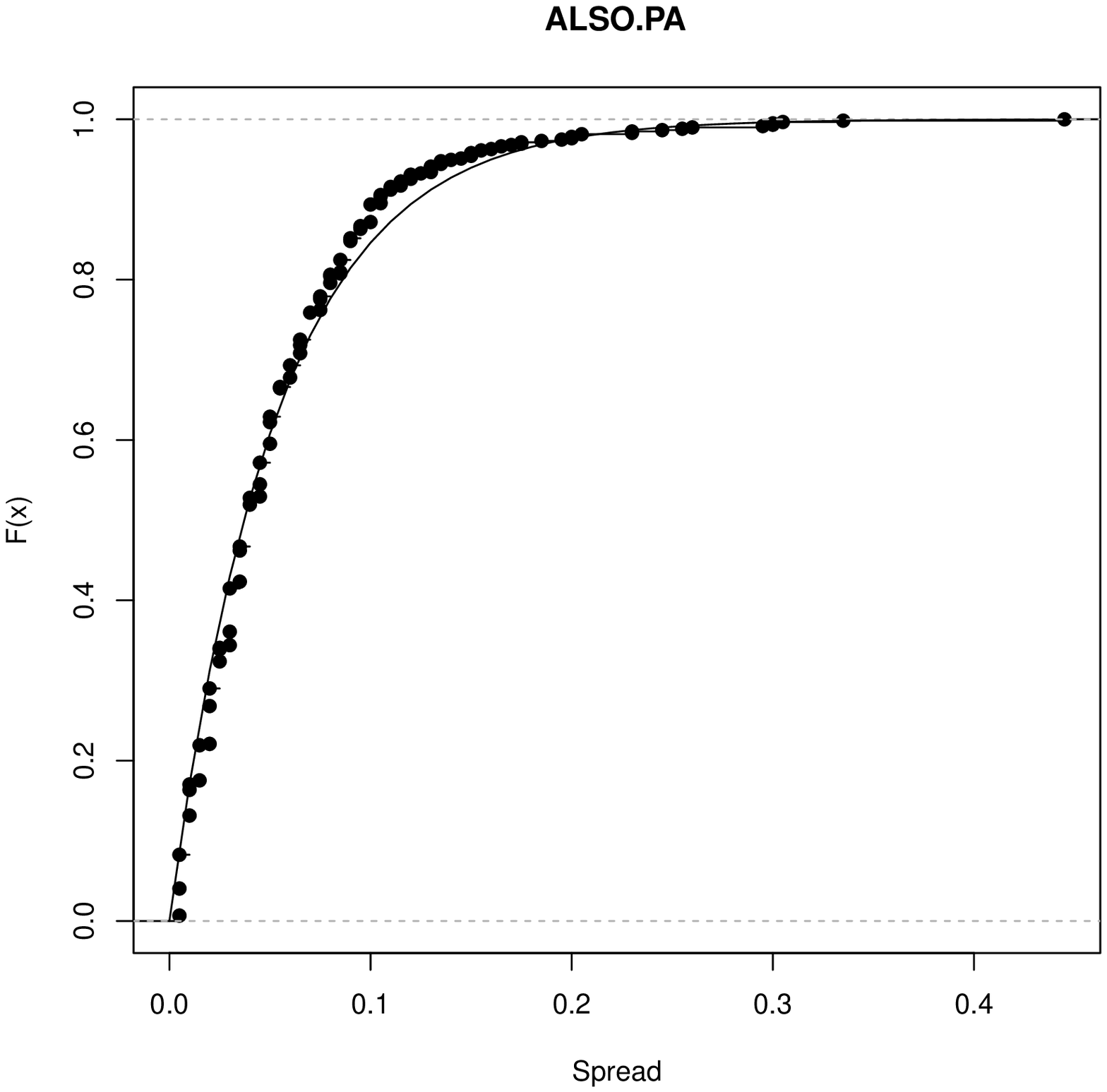}
	& \includegraphics[width=0.3\textwidth]{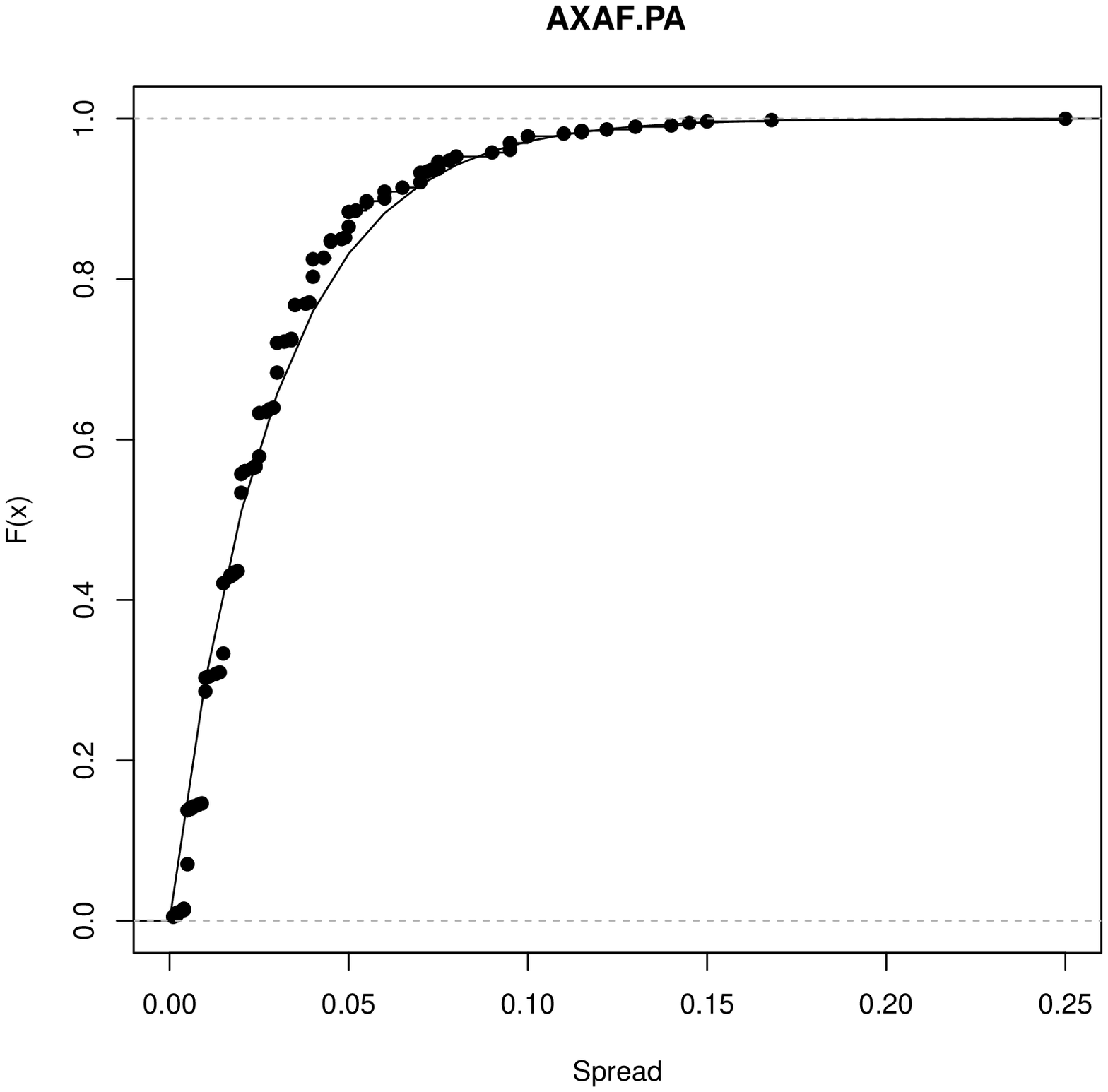}
	\\ \includegraphics[width=0.3\textwidth]{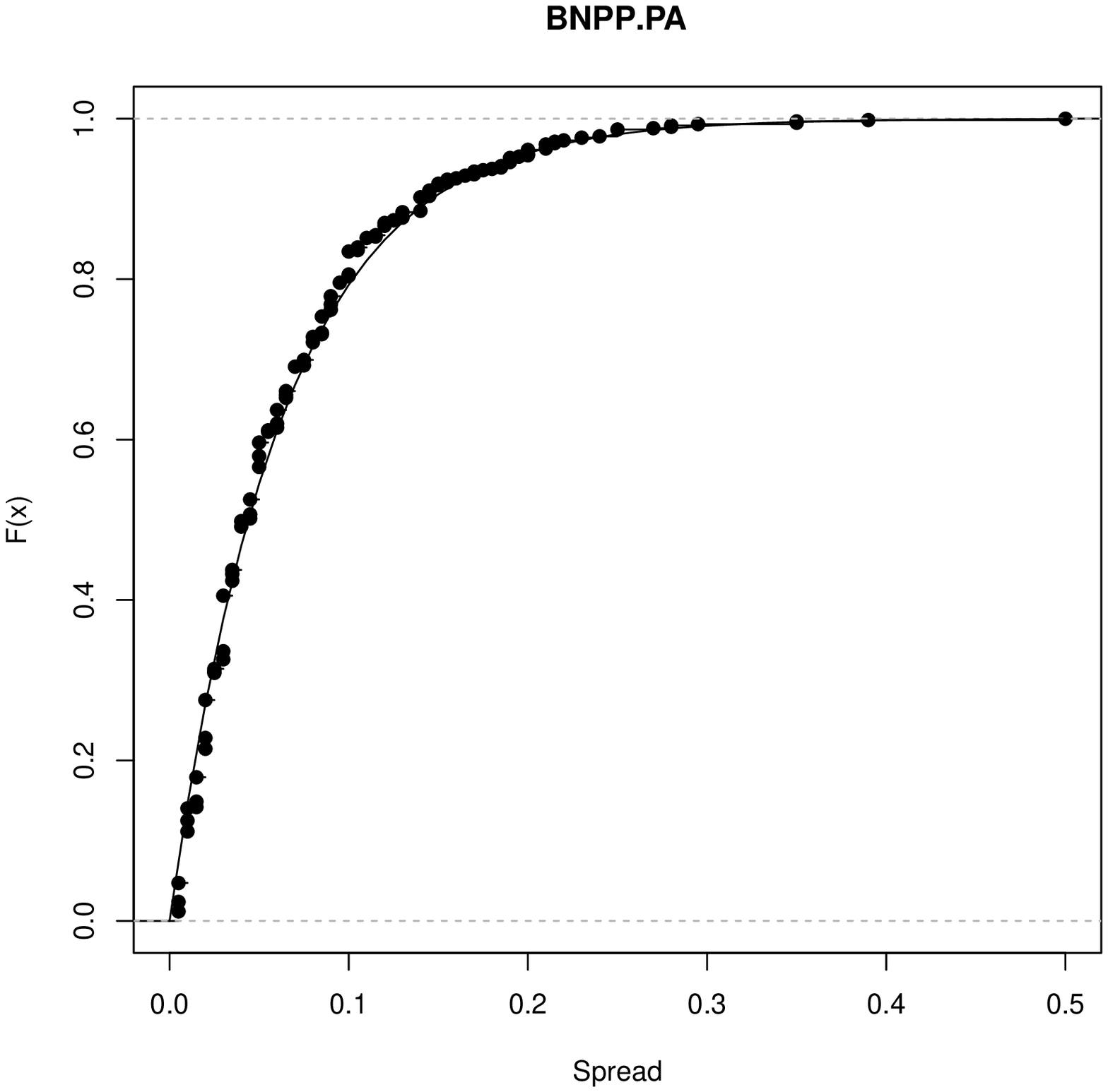}
	& \includegraphics[width=0.3\textwidth]{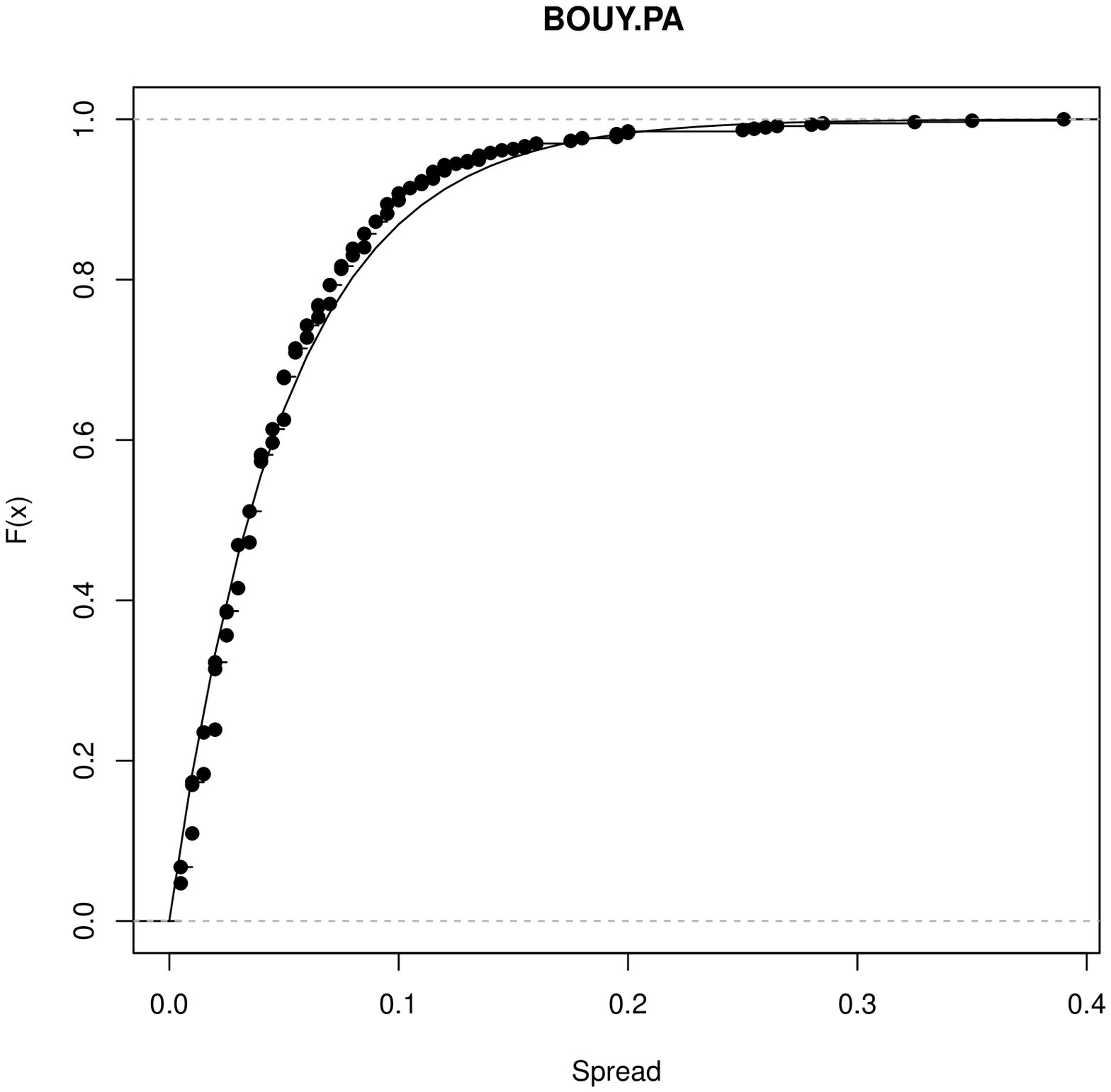}
	& \includegraphics[width=0.3\textwidth]{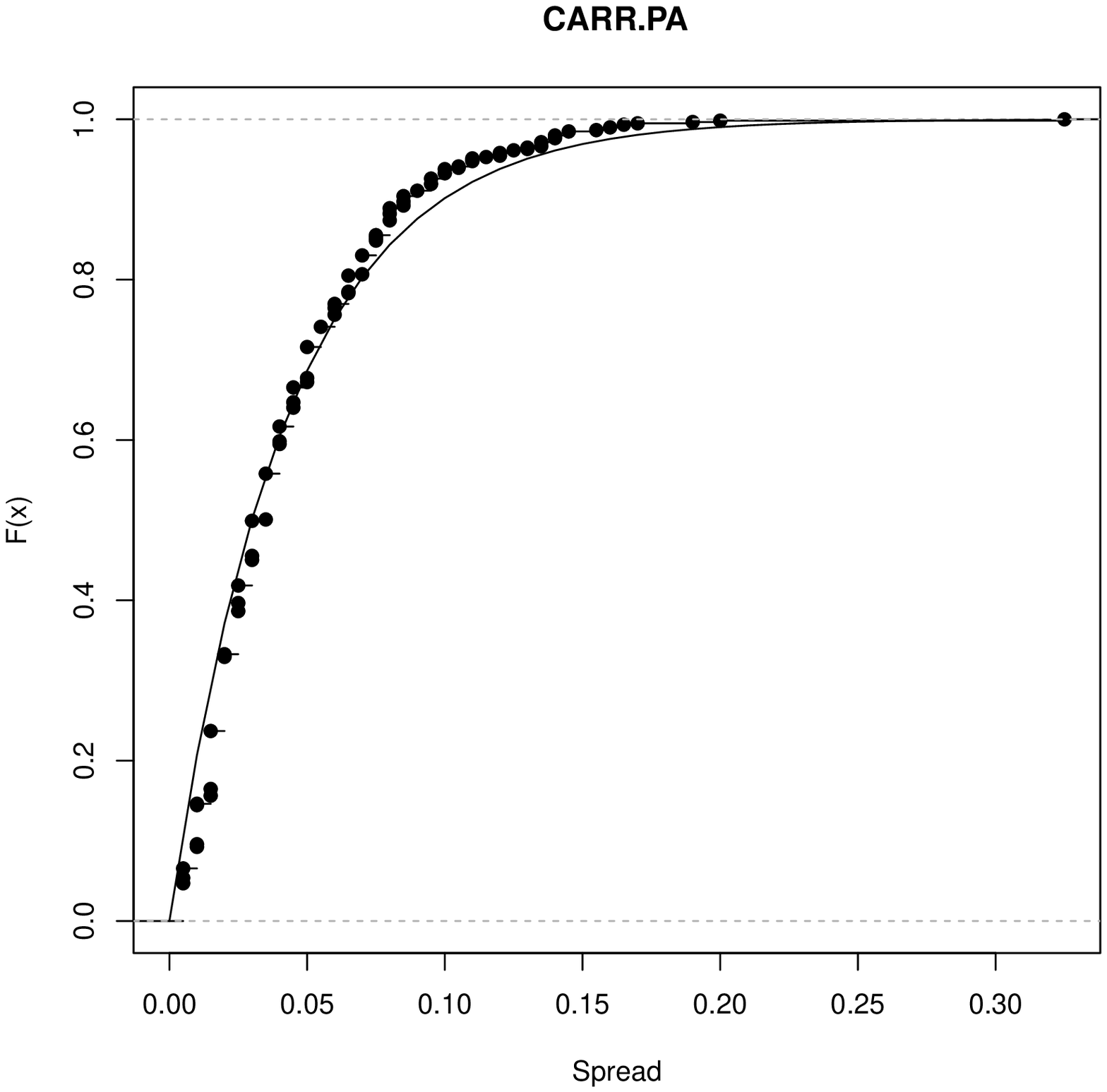}
	\\ \includegraphics[width=0.3\textwidth]{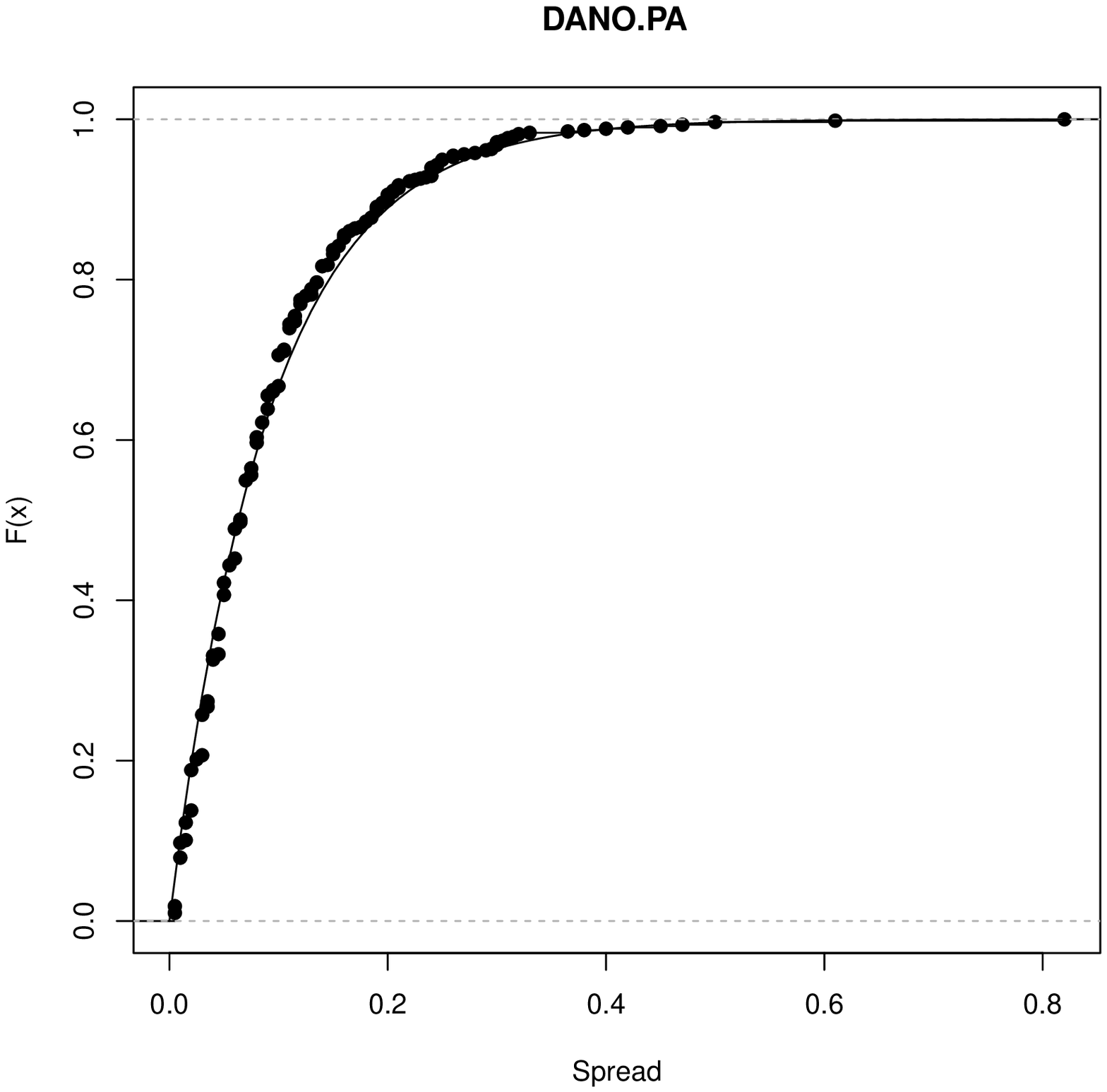}
	& \includegraphics[width=0.3\textwidth]{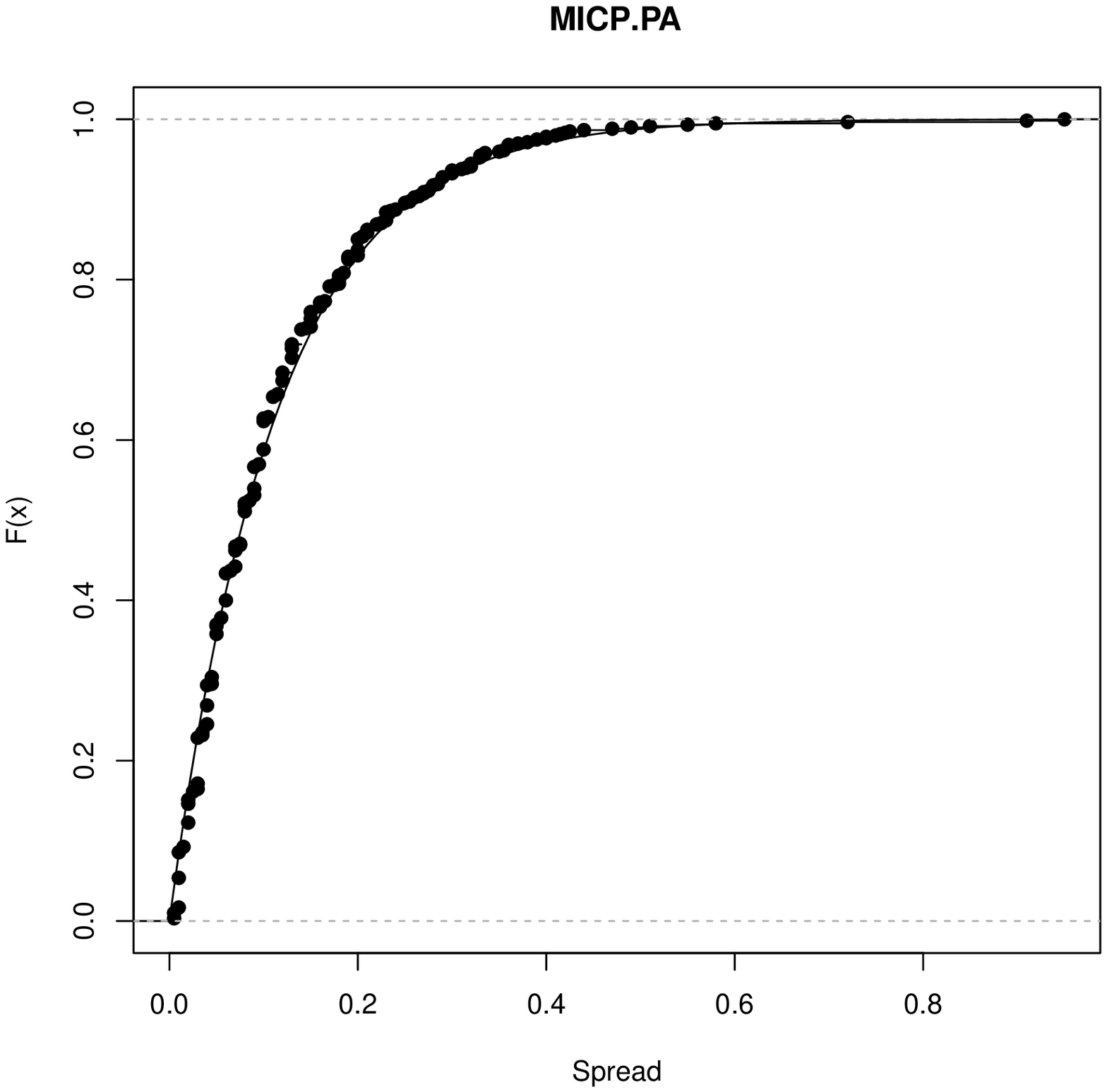}
	& \includegraphics[width=0.3\textwidth]{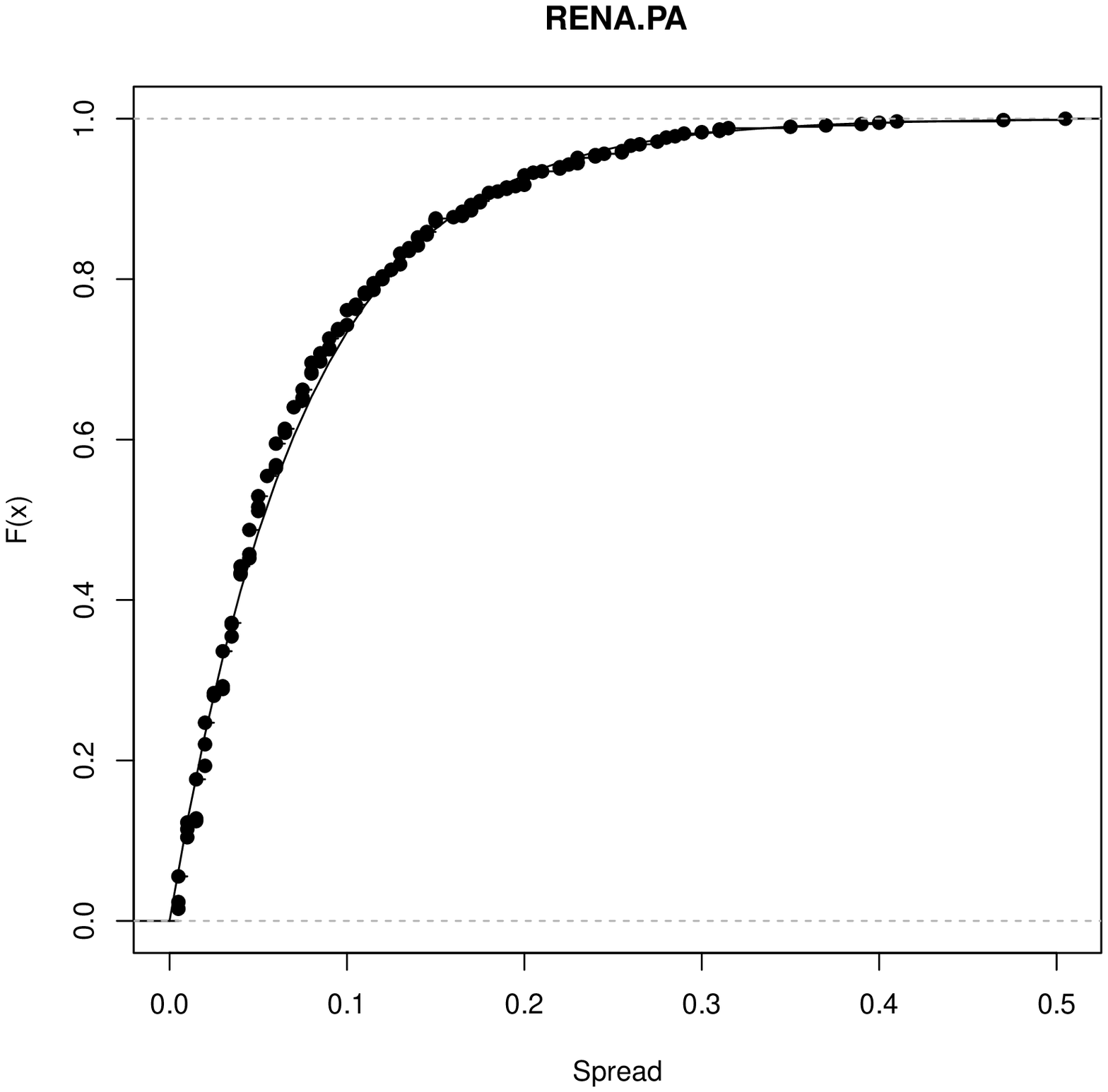}
	\\ \includegraphics[width=0.3\textwidth]{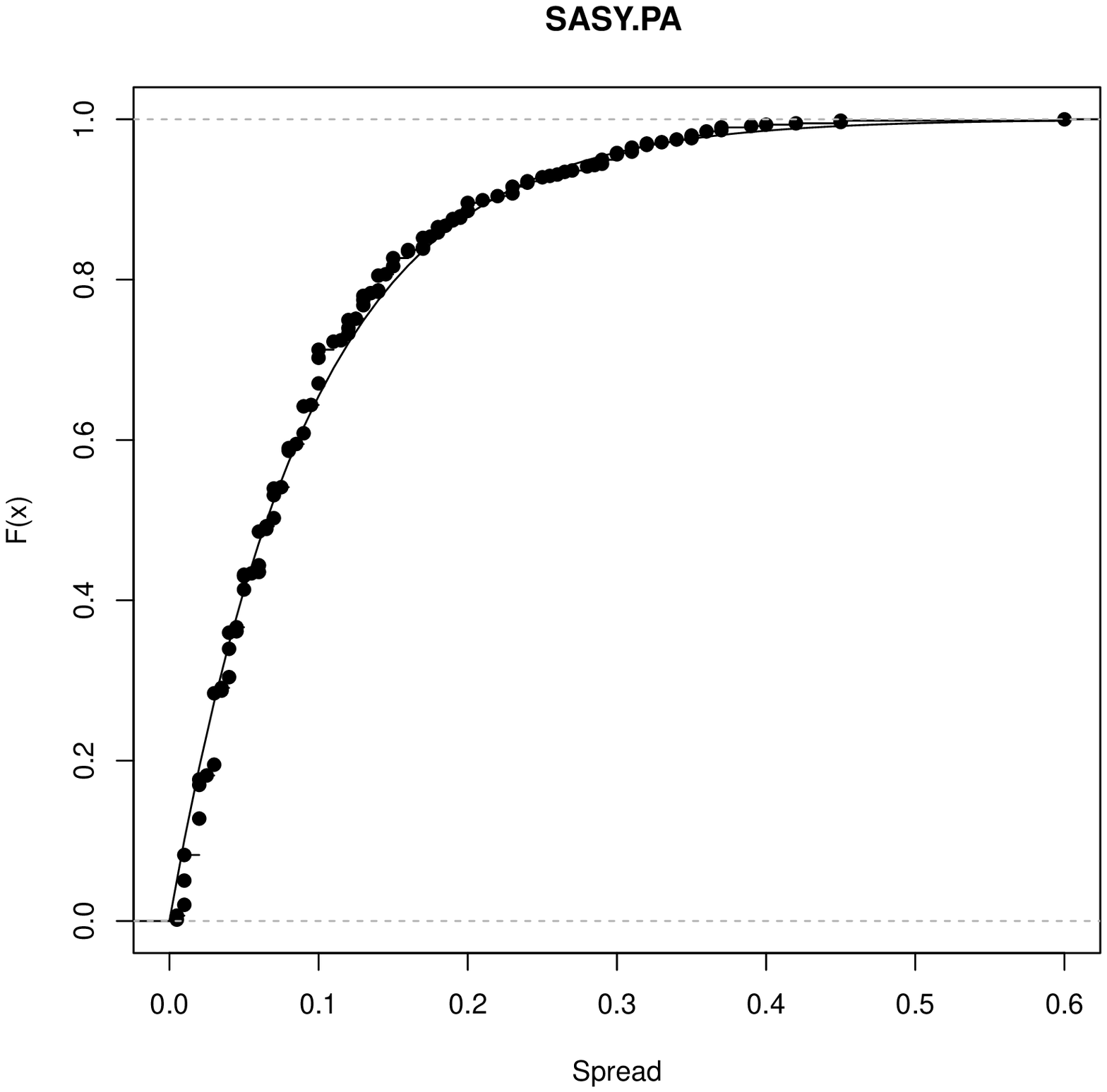}
	& \includegraphics[width=0.3\textwidth]{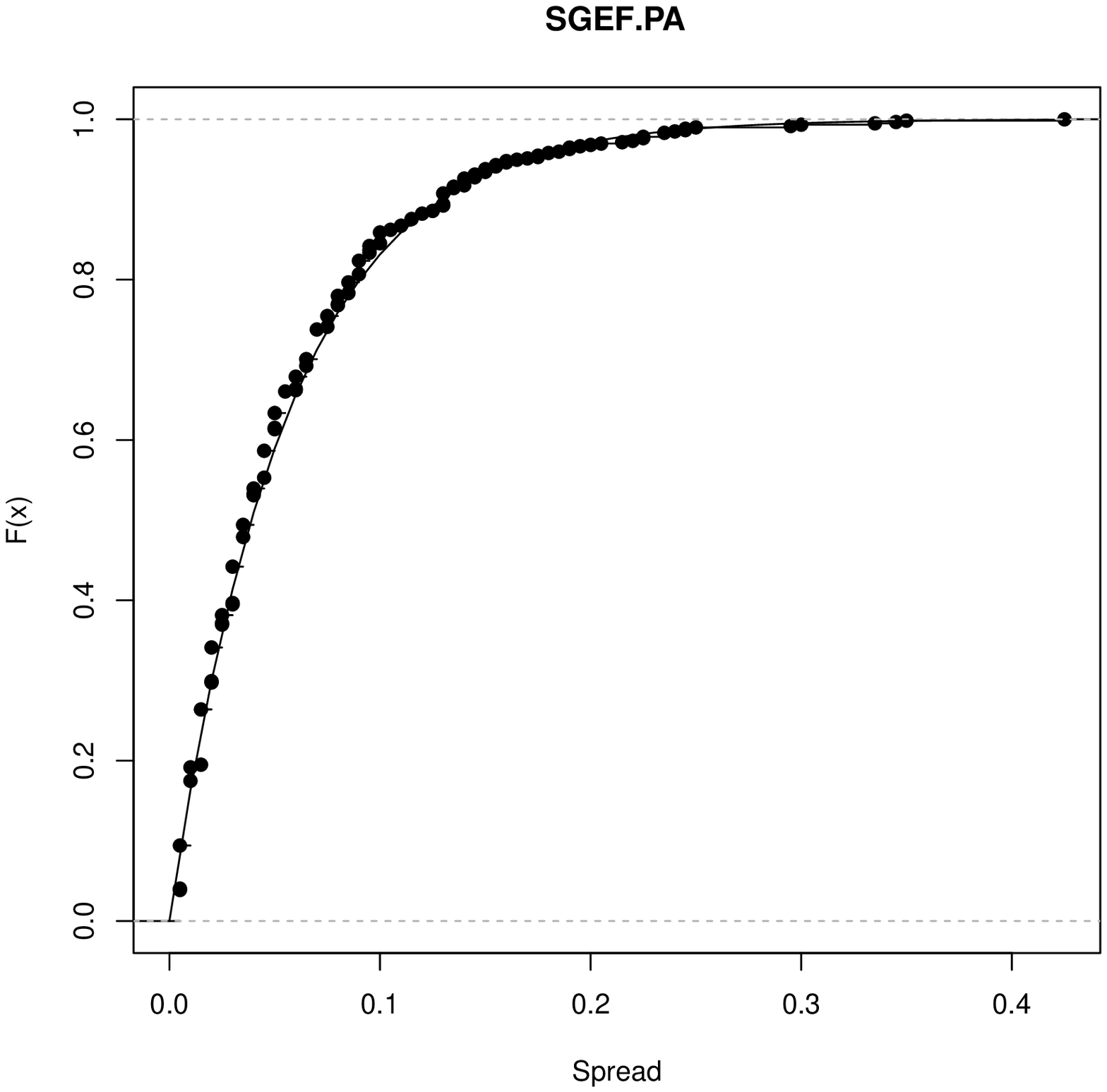}
	& \includegraphics[width=0.3\textwidth]{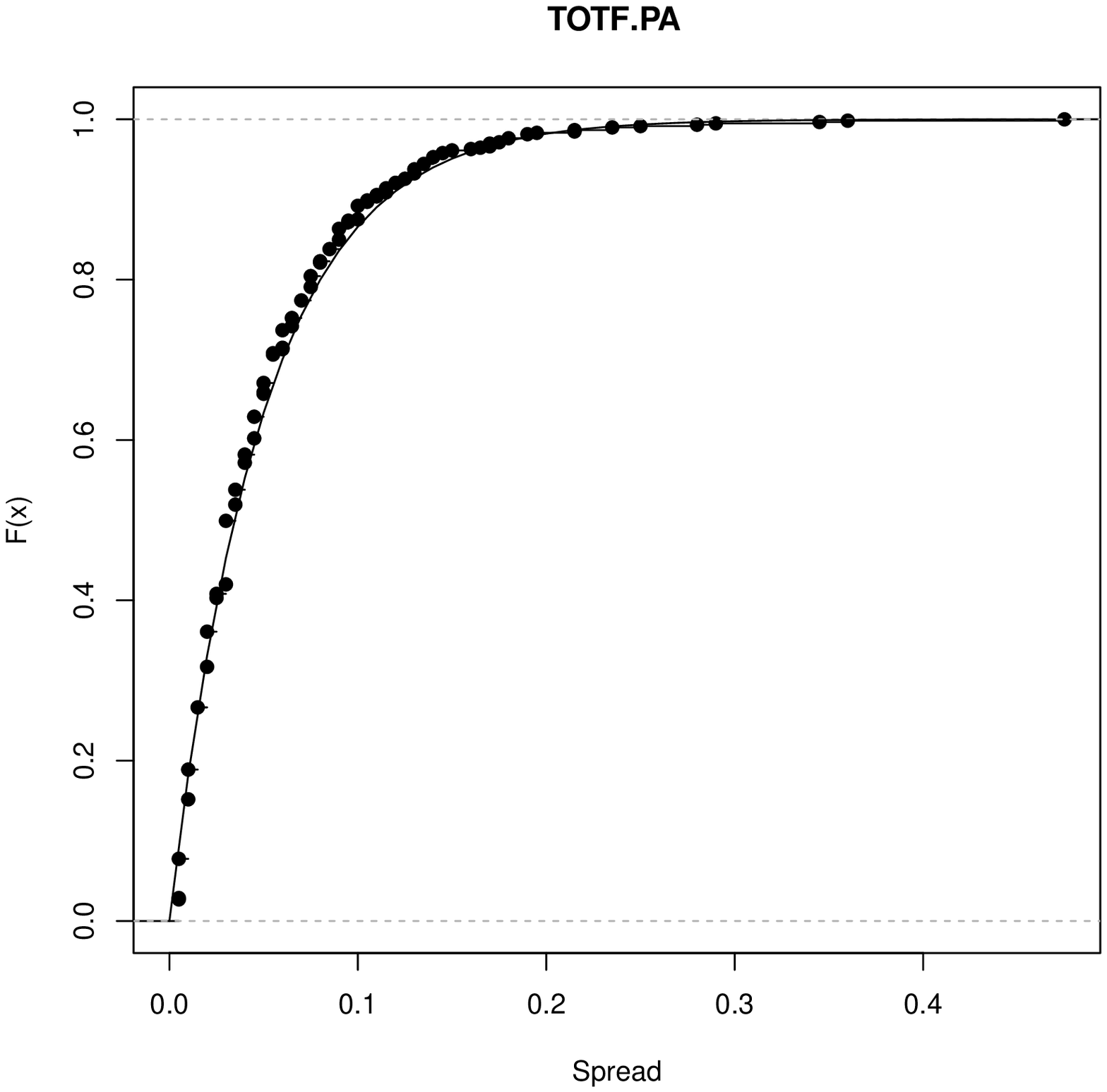}
\end{tabular}
\caption{Empirical cumulative distribution function of the post-clearing spread computed for $12$ CAC 40 stocks traded on the Paris Stock Exchange from March 2011 to June 2013. Full line is the exponential MLE fit.}
\label{fig:SpreadExpFit-DistributionFunction}
\end{figure}

\begin{figure}
\centering
\begin{tabular}{ccc}
	\includegraphics[width=0.3\textwidth]{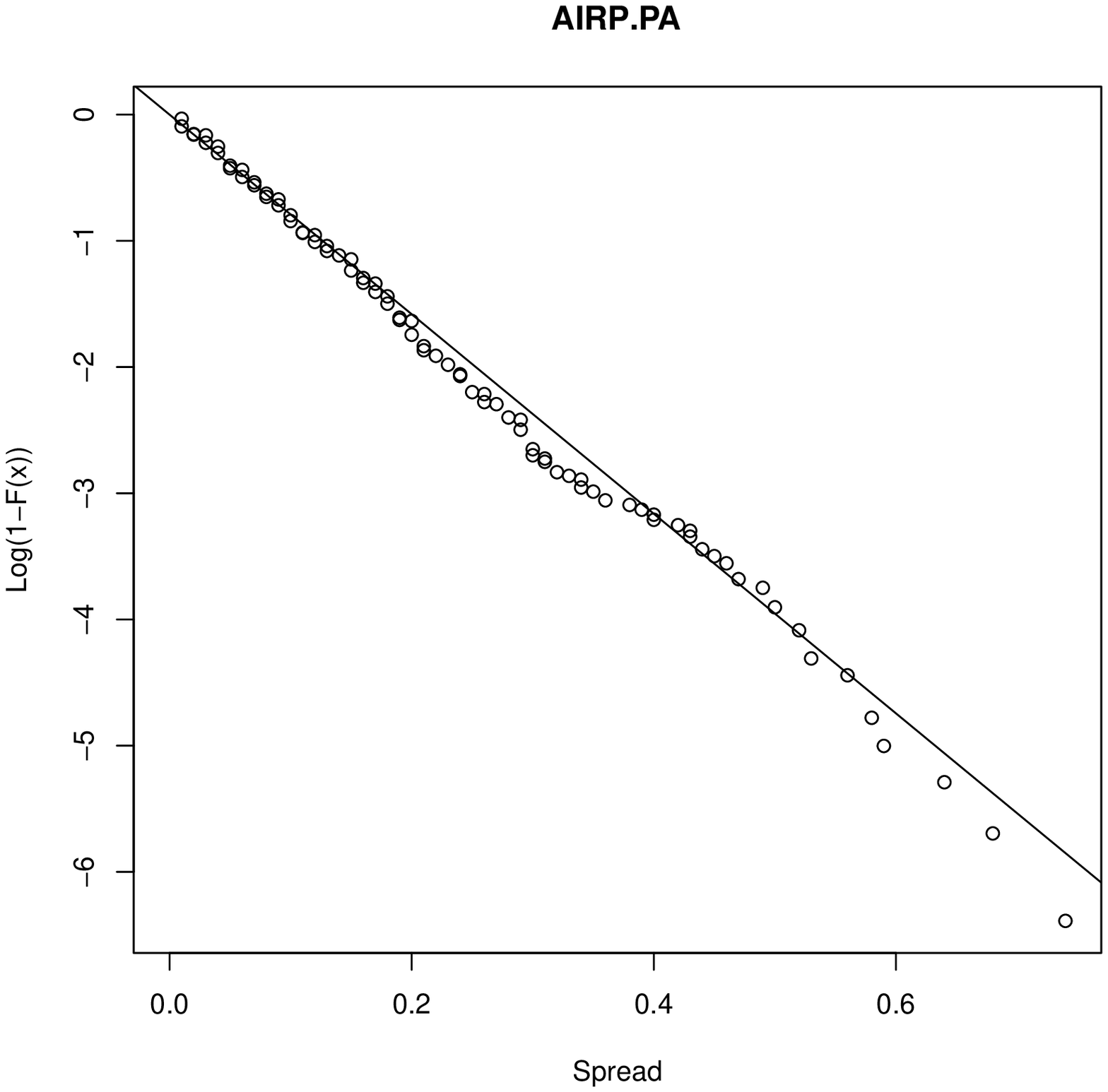}
	& \includegraphics[width=0.3\textwidth]{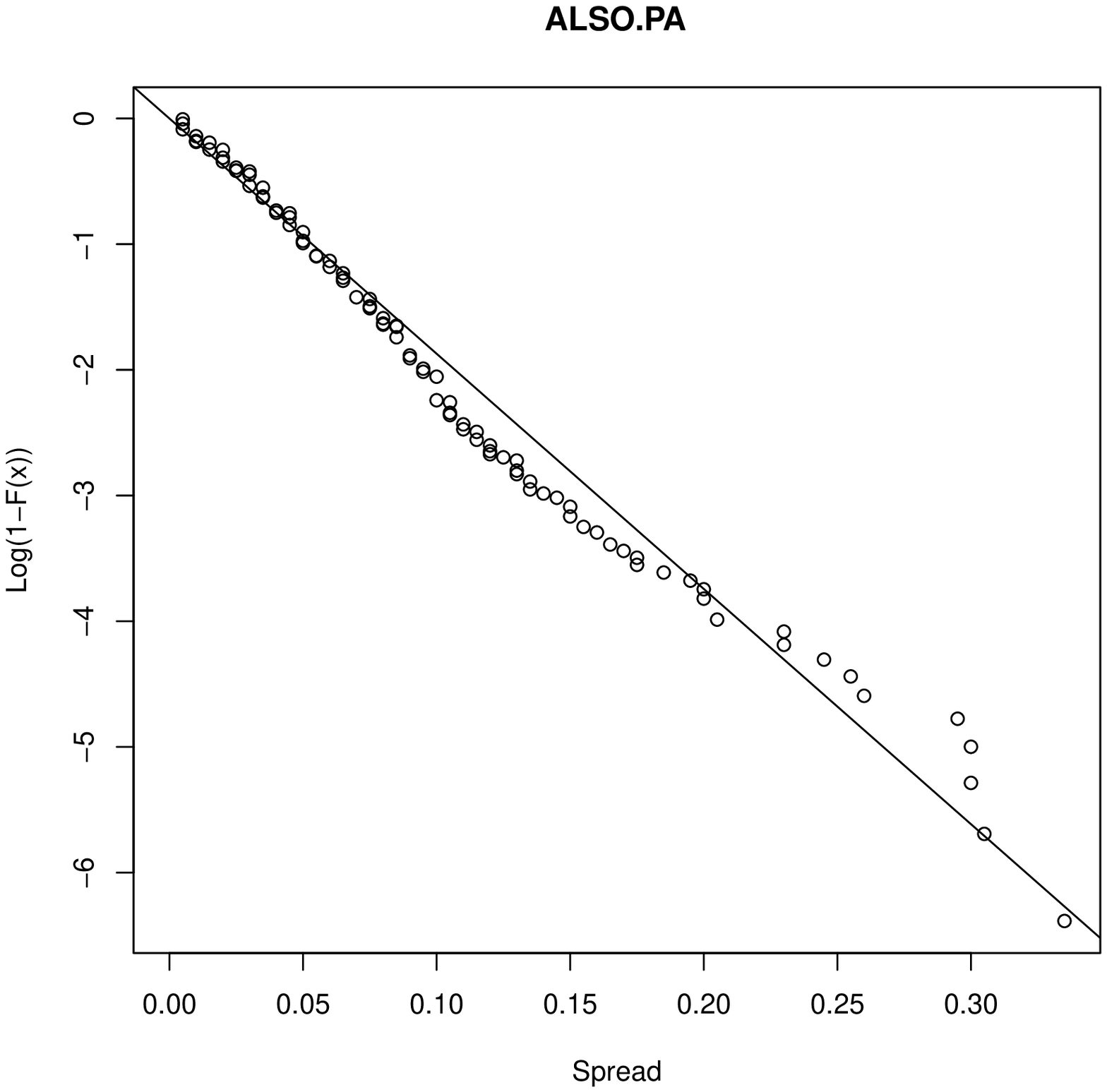}
	& \includegraphics[width=0.3\textwidth]{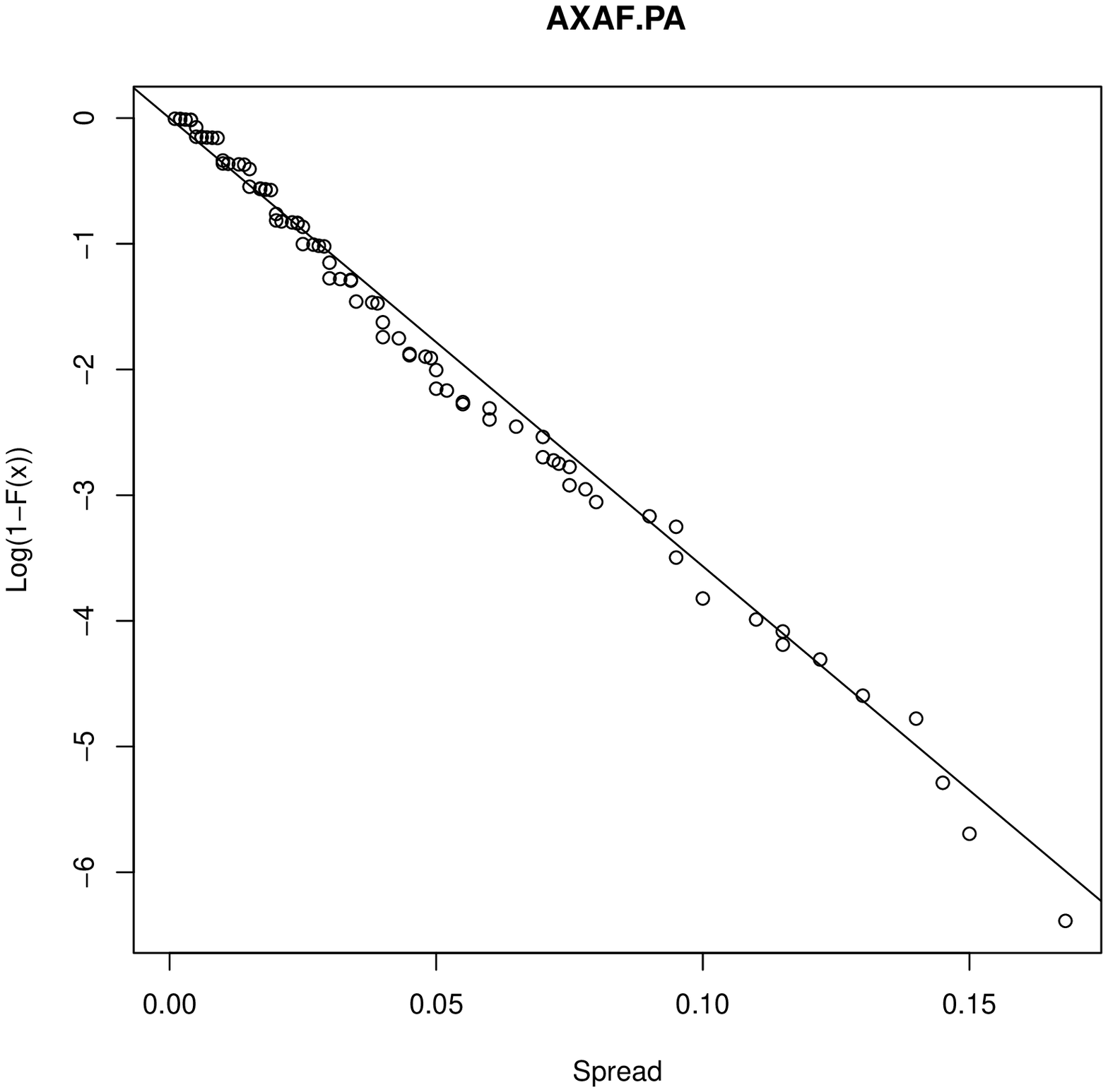}
	\\ \includegraphics[width=0.3\textwidth]{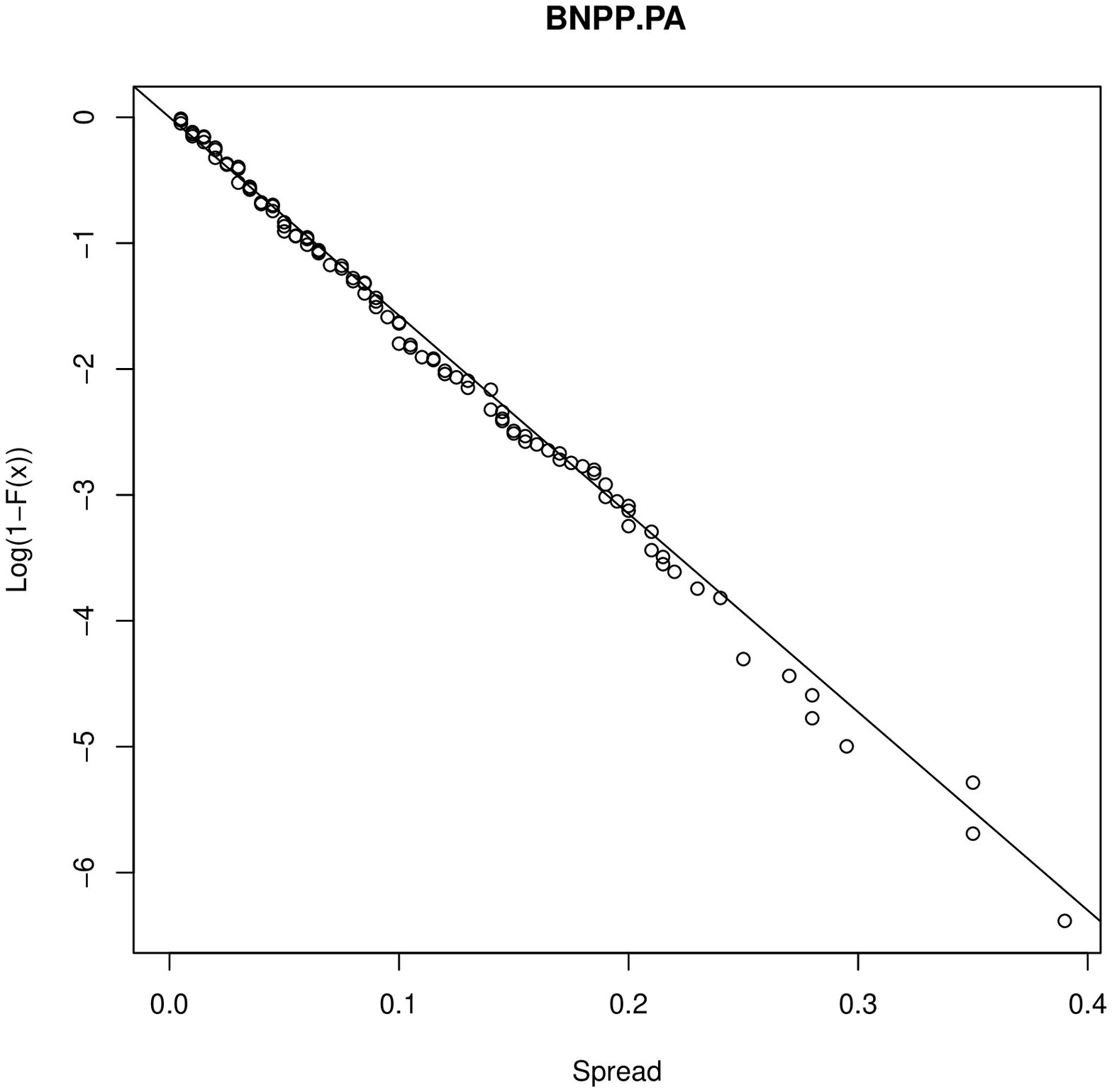}
	& \includegraphics[width=0.3\textwidth]{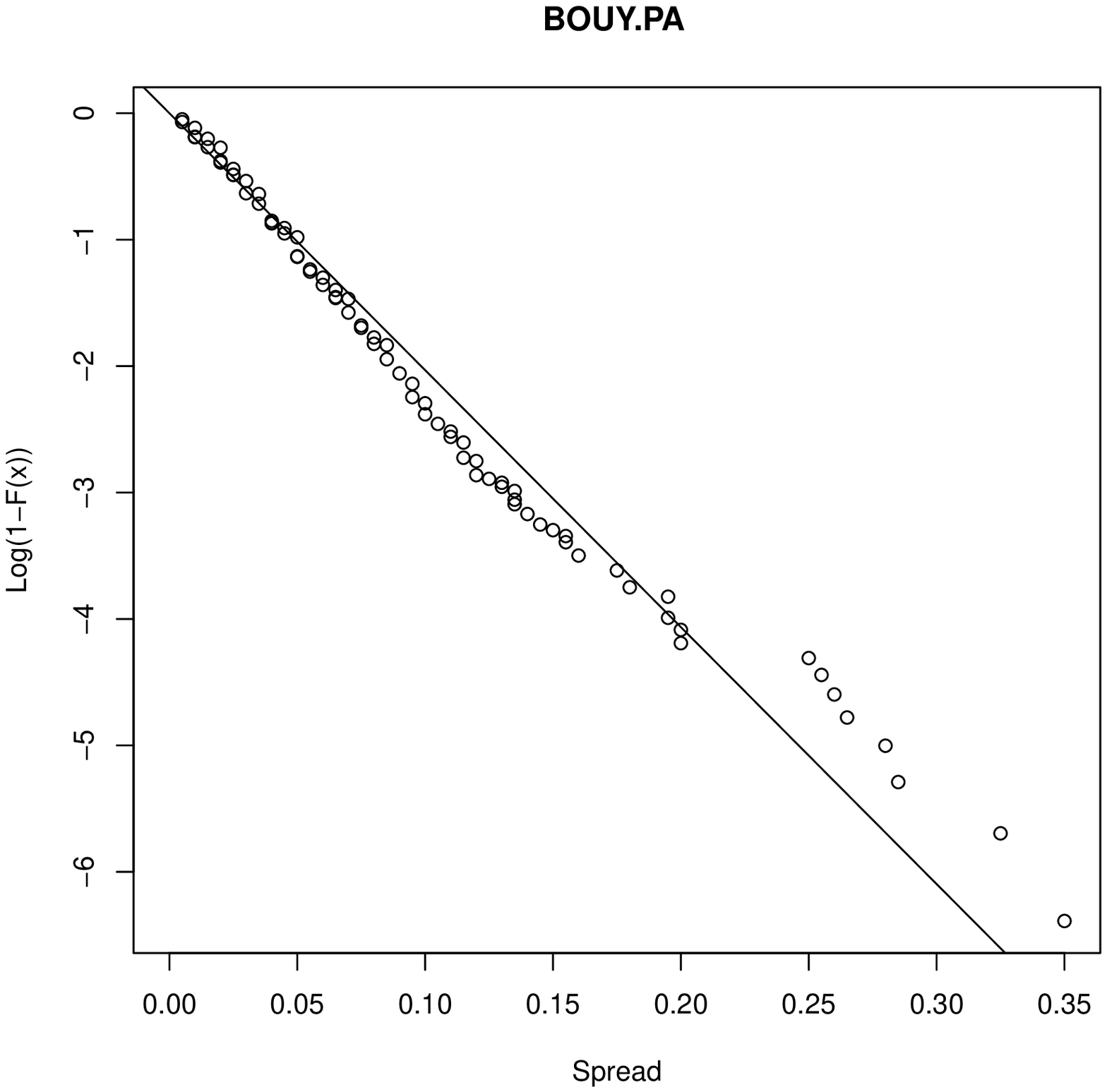}
	& \includegraphics[width=0.3\textwidth]{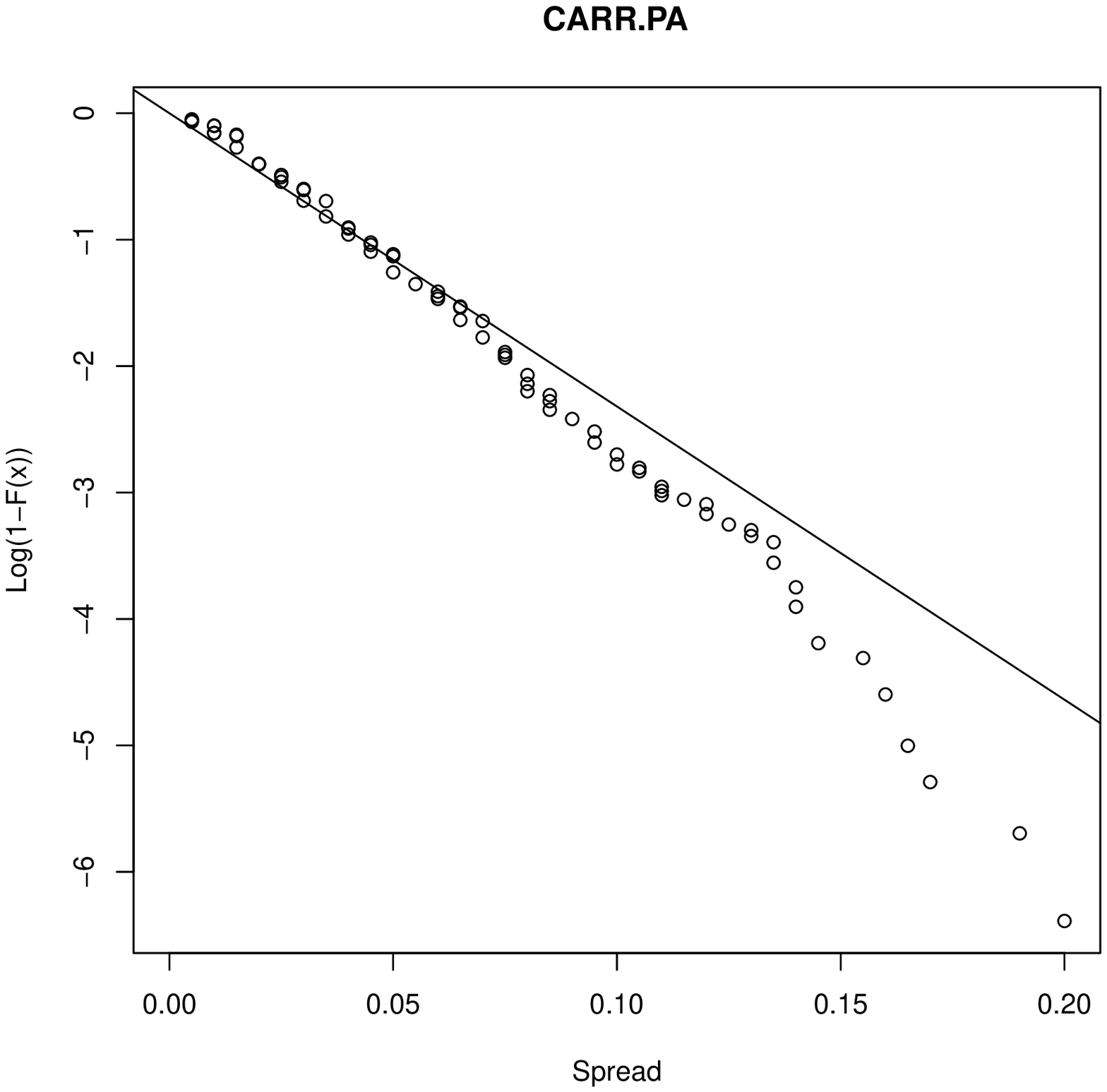}
	\\ \includegraphics[width=0.3\textwidth]{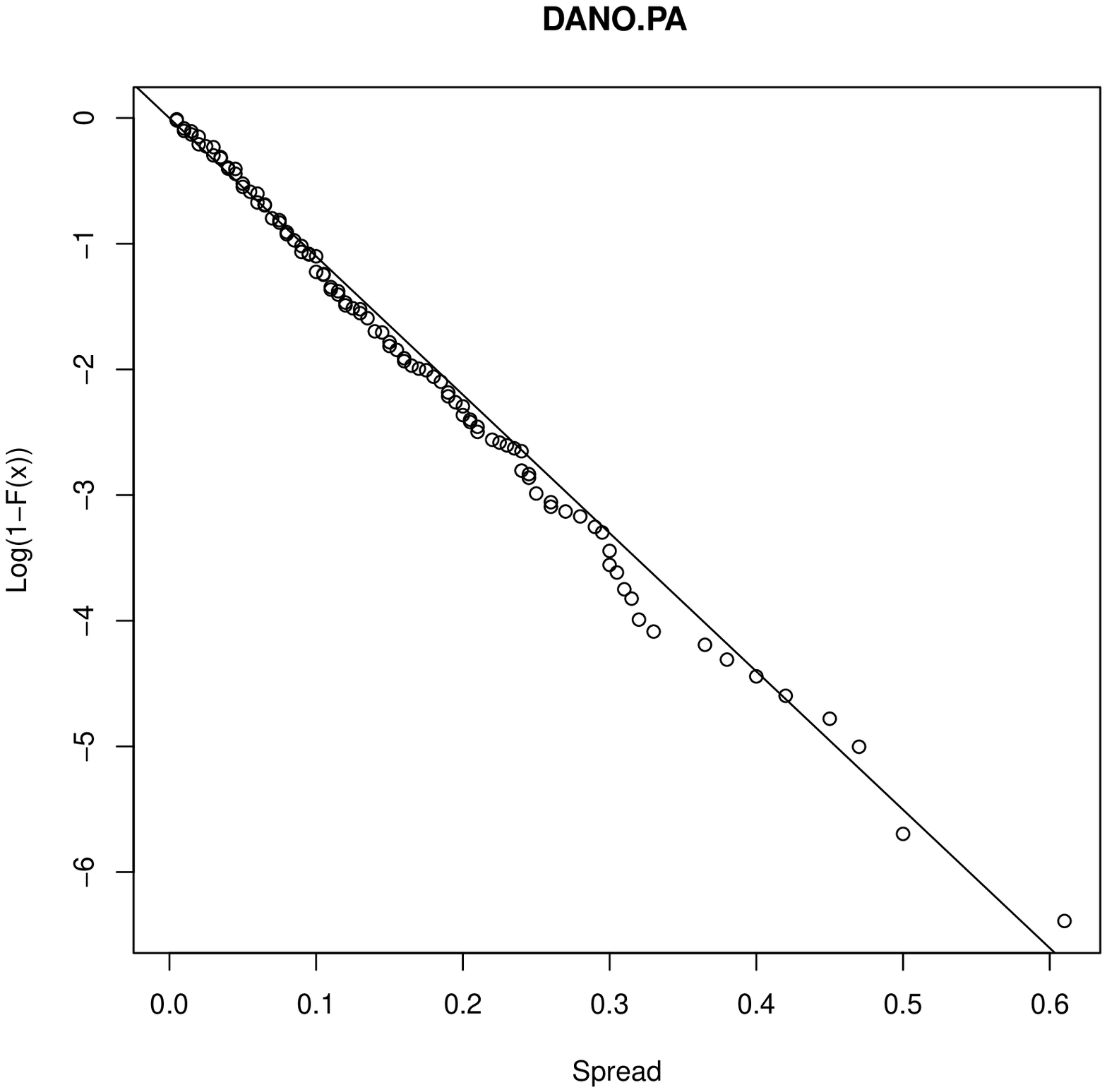}
	& \includegraphics[width=0.3\textwidth]{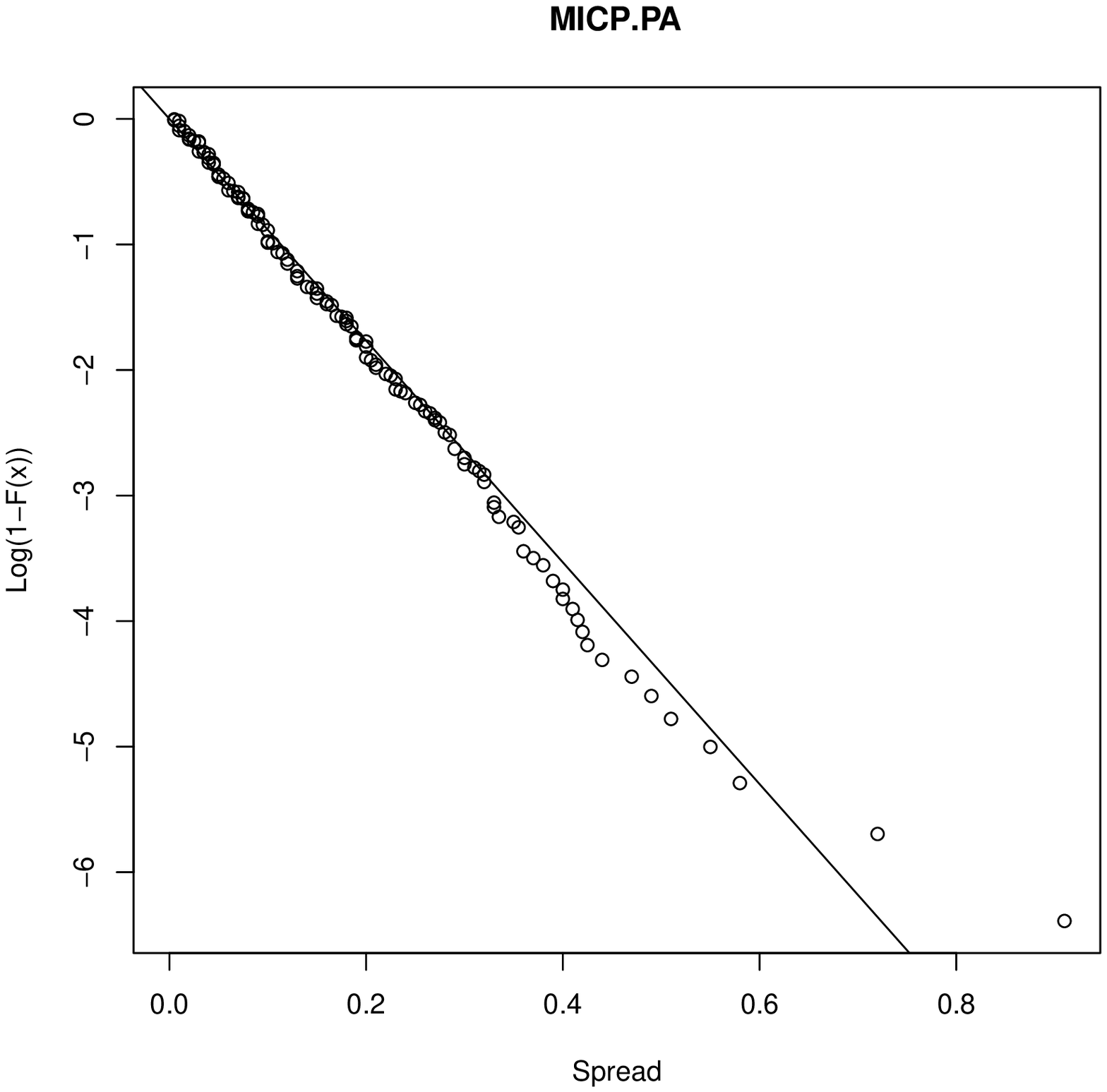}
	& \includegraphics[width=0.3\textwidth]{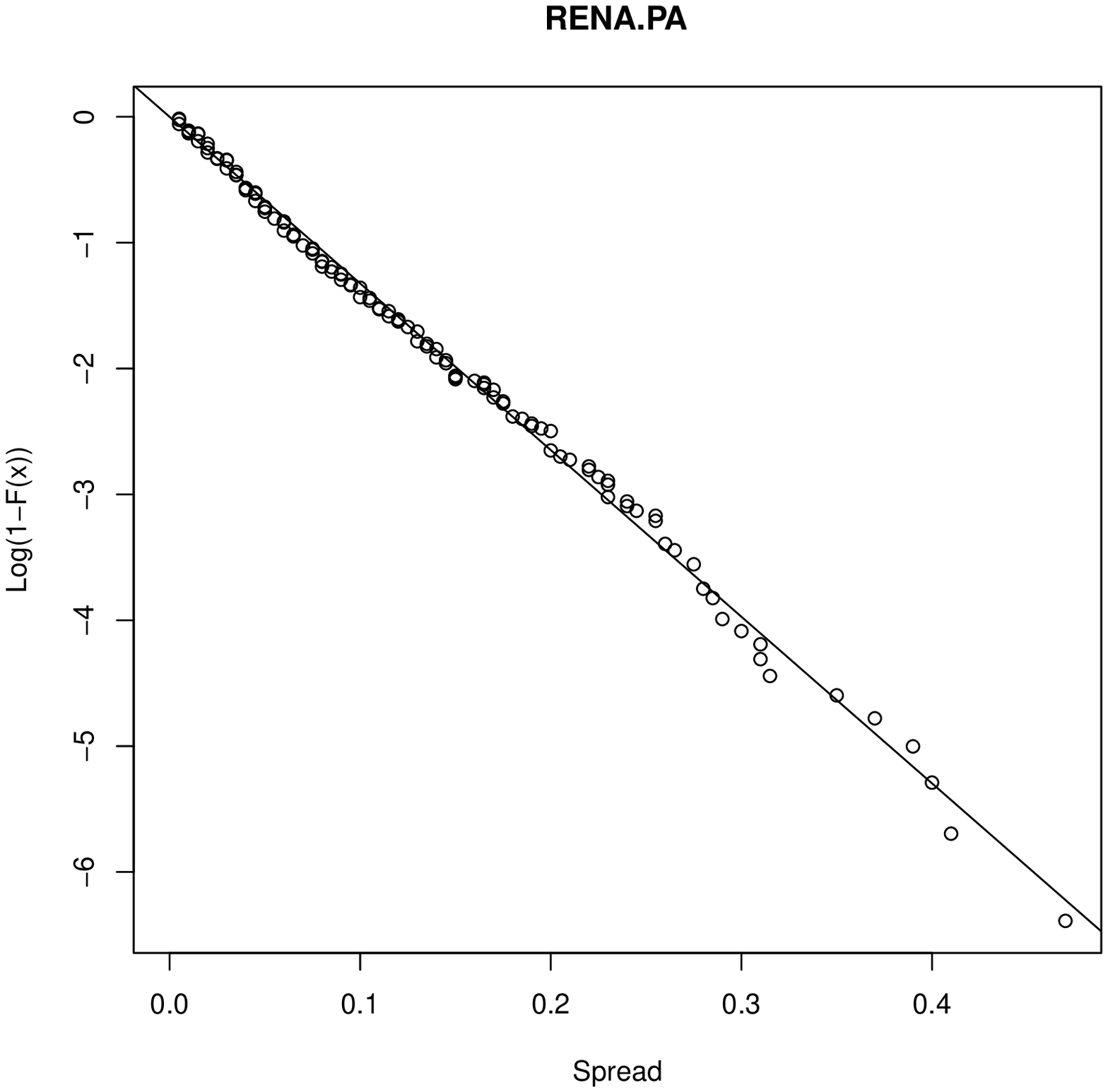}
	\\ \includegraphics[width=0.3\textwidth]{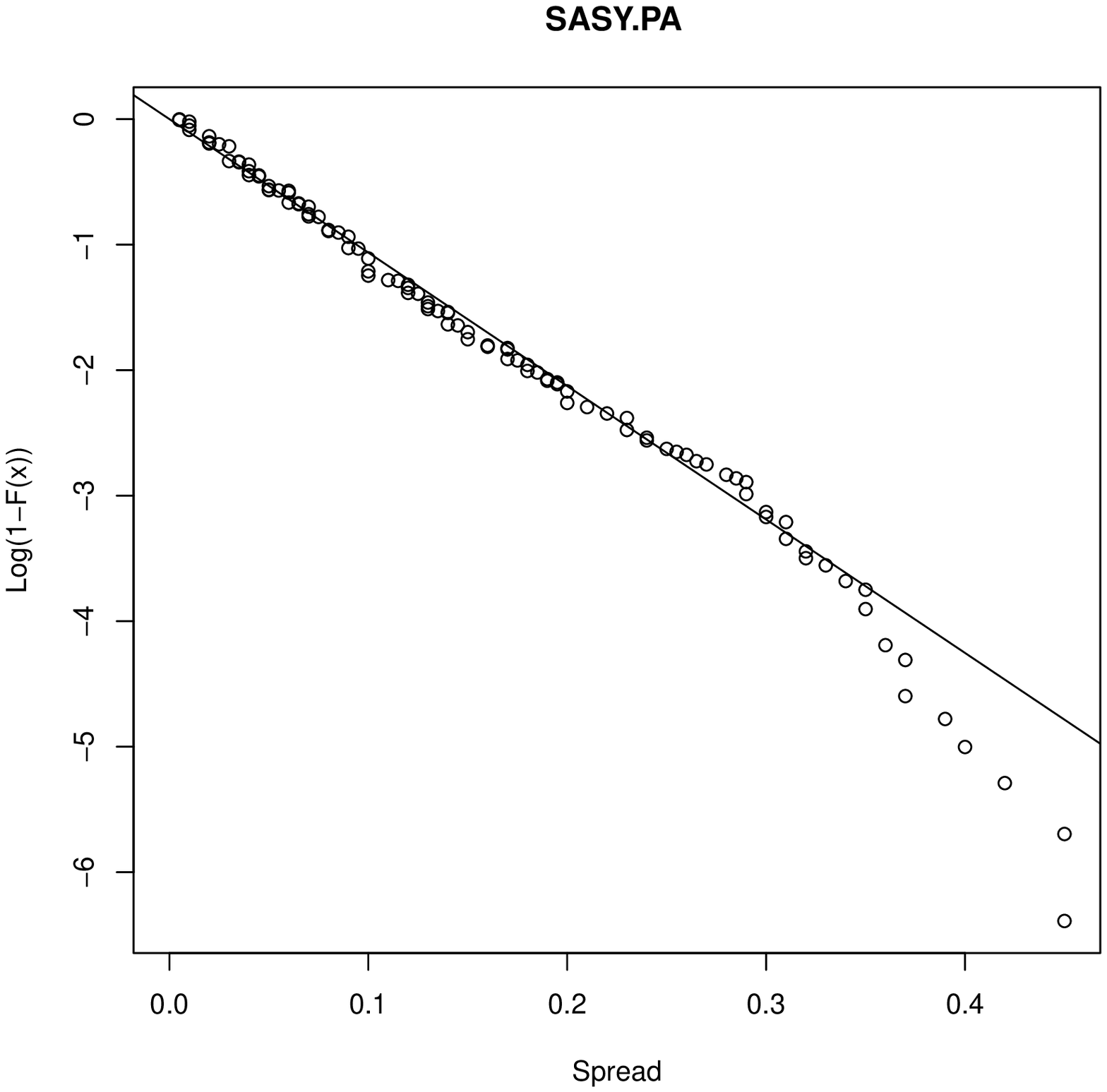}
	& \includegraphics[width=0.3\textwidth]{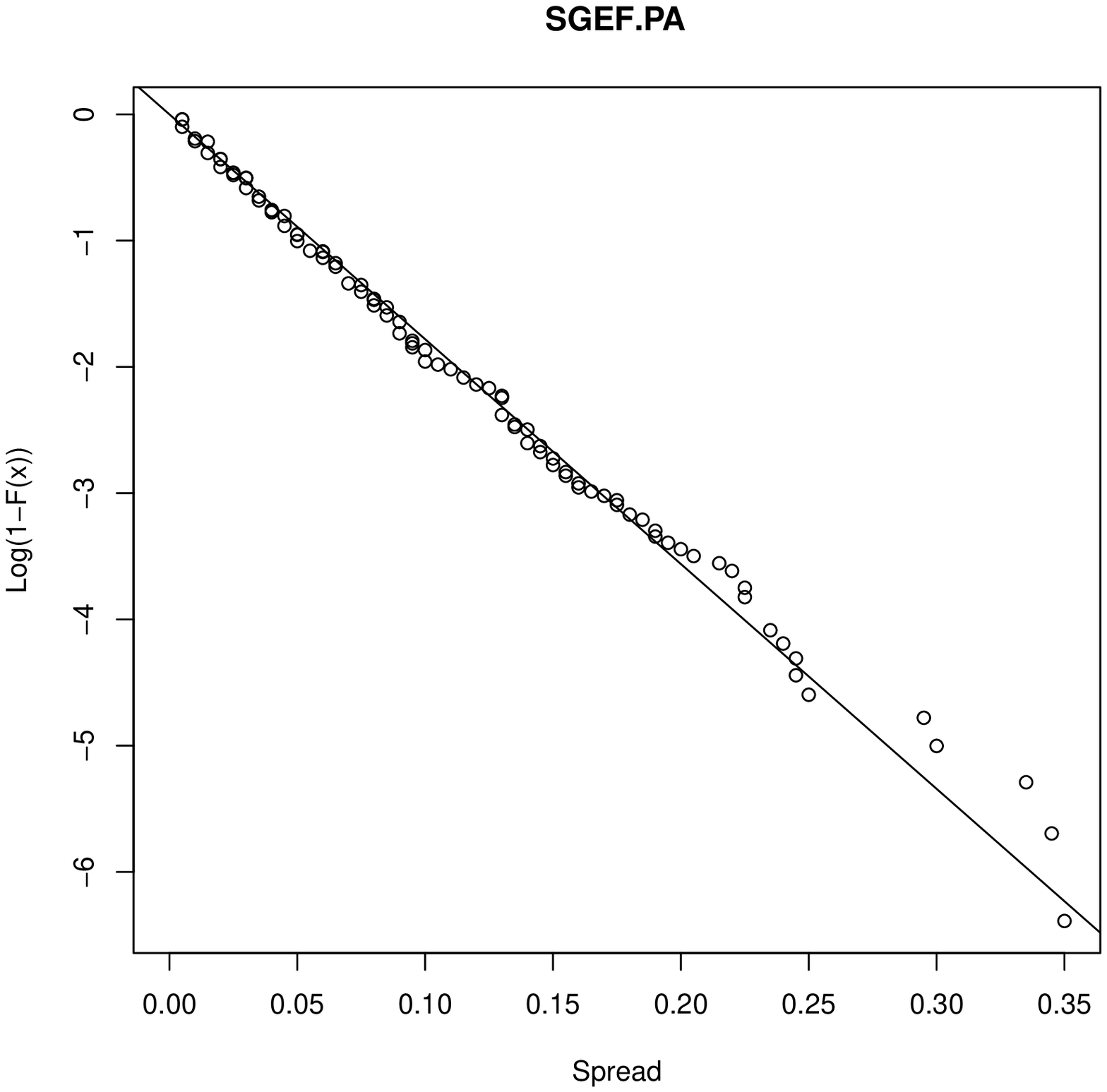}
	& \includegraphics[width=0.3\textwidth]{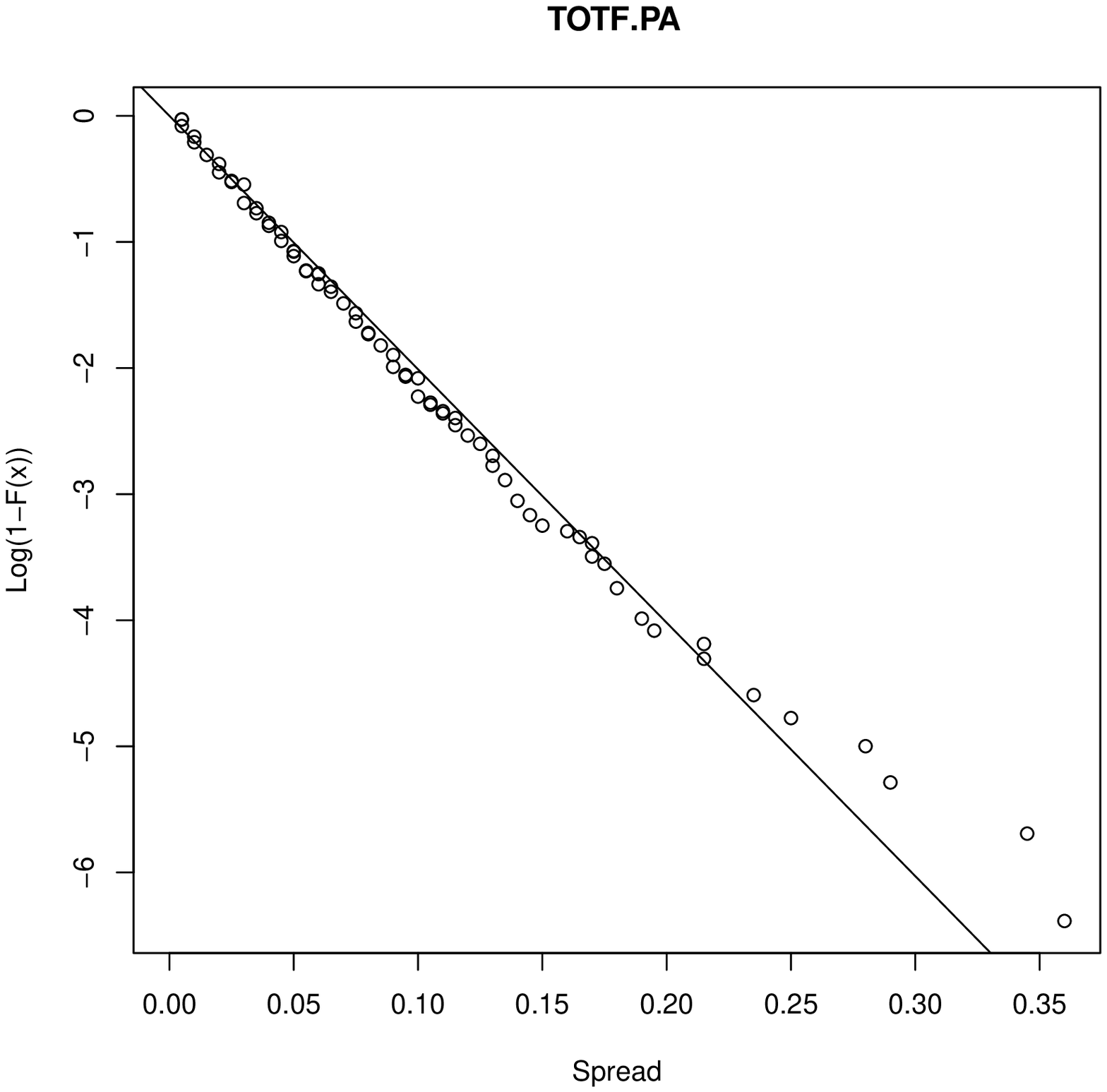}
\end{tabular}
\caption{Logarithm of the empirical survival function of the distribution of the post-clearing spread computed for $12$ CAC 40 stocks traded on the Paris Stock Exchange from March 2011 to June 2013. Straight line is the exponential MLE fit.}
\label{fig:SpreadExpFit-SurvivalFunction}
\end{figure}

These empirical observations indicate that the exponential distribution predicted for the clearing price range by the basic Poisson model might indeed be an acceptable approximation of the distribution of the spread computed by the clearing mechanism of the opening call auction.
As a future work, the next empirical step would be to test whether a valid link can be establish between the parameter of this exponential distribution and the parameters defining the flows of submitted orders during the call, in order to (in)validate the scaling suggested in proposition \ref{prop:AsymptoticPriceRange}. Note however that this will require a completely different set of data, since all the order book data (not just the bid and ask quotes) during all the pre-opening call period (not just the opening values) is needed to estimate the parameters of the order flow.

\section{Allowing for the cancellation of trading orders}
\label{sec:Cancellation}

We have until now presented results in a call auction model where submitted orders cannot be cancelled. This has allowed us easy comparisons with previous results.
It is however simple to add a cancellation mechanism. 
Let us assume from now on that any submitted order may be cancelled between the moment of its submission and the closing of the call auction at $T$. Let us assume furthermore that the lifetimes of orders (i.e. the time interval between their submission and their potential cancellation) form a set of independent random variables distributed exponentially with parameter $\theta_A$ for ask orders, and $\theta_B$ for bid orders.

Then the number of ask orders at time $T$ is distributed as the value of a birth-and-death process with constant birth rate $\lambda_A$ and linear death rate $n\theta_A$.
If $p_n(t)$ denotes the probability that the process is equal to $n$ at time $t$, then the set of Kolmogorov forward equations is written for all $n\in\mathbb N^*$ \citep[see e.g.][chap. 8]{Bremaud1999}:
\begin{equation}
	\frac{dp_n(t)}{dt} = -(\lambda_A+n\theta_A)p_n(t) + \lambda p_{n-1}(t) \mathbf 1_{n\geq 1} + (n+1)\theta_A p_{n+1}(t).
\end{equation}
Classically, this is explicitly solvable through the use of the moment generating function, and we obtain that the number of ask orders at time $T$ is distributed according to a Poisson distribution with parameter $\frac{\lambda_A}{\theta_A}\left(1-e^{-\theta_A T}\right)$. A similar result is obviously obtained for the bid orders.

Therefore, all the developments of the previous sections are applicable to the general model with cancellation. It suffices to substitute $\lambda$ and $\alpha$ in the previous formulas from proposition \ref{fig:TradedVolumeDistribution} to proposition \ref{fig:PriceRangeConvergence}, with respectively $\lambda'$ and $\alpha'$ defined by :
\begin{align}
	\lambda' & = \frac{\lambda}{T}\left(
		\frac{\alpha}{\theta_A}(1-e^{-\theta_A T}) + \frac{1-\alpha}{\theta_B}(1-e^{-\theta_B T})
	\right)
	\\
	\alpha' & = \frac{1}{1 + \left(\frac{1}{\alpha}-1\right)
	\frac{\theta_A}{\theta_B} \frac{1-e^{-\theta_B T}}{1-e^{-\theta_A T}}}  
\end{align}
Obviously, if $T\to 0$ or equivalently if $\theta_A\to 0$ and $\theta_B\to 0$, then the cancellation mechanism is negligible and $\alpha'\to\alpha$ and $\lambda'\to\lambda$.
Note also that in the case $\theta_A=\theta_B=\theta$, then both bid and ask are modified the same way, and as expected $\alpha'=\alpha$, i.e. market imbalance is unchanged.
However, if $\theta_A$ and $\theta_B$ are different, then the cancellations may either increase or lessen the imbalance of the market.

The formulas obtained are in accordance with the expected effect of the cancellation mechanism on the call auction : as $\theta_A$ and $\theta_B$ increases, less orders enter the clearing process at time $T$, and therefore the traded volume decreases, while the clearing price range and the variance of the clearing prices increase.
More precisely, the scaling obtained is roughly proportional. Let us simplify the formulas to follow by assuming that bid and ask orders have the same expected lifetime $\theta_A^{-1}=\theta_B^{-1}=\theta^{-1}$.
In the liquid case where $\lambda\to+\infty$, the asymptotic distribution of the traded volume is then:
\begin{equation}
	\mathcal N\left(\frac{\lambda}{\theta}(1-e^{-\theta T})\alpha(1-\alpha)
	,
	\sqrt{ \frac{\lambda}{\theta}(1-e^{-\theta T}) \alpha(1-\alpha)(1-2\alpha(1-\alpha))}
	 \right),
\end{equation}
i.e. the mean traded volume is roughly inversely proportional to $\theta$.
The asymptotic distribution of the clearing price is:
\begin{equation}
	\mathcal N\left( F^{-1}(1-\alpha)
	,
	\frac{1}{f(F^{-1}(1-\alpha))}
	\sqrt{\frac{2\theta\alpha(1-\alpha)}{\lambda(1-e^{-\theta T})}}
	 \right),
\end{equation}
and finally, the asymptotic distribution of the clearing price range is:
\begin{equation}
	\mathcal E\left(\frac{\lambda}{\theta} (1-e^{-\theta T}) f(F^{-1}(1-\alpha)) \right),
\end{equation}
i.e. both the variance of the prices and the average clearing price range are roughly proportional to $\theta$, the inverse of the expected lifetimes of orders.

\section{Conclusion}
\label{sec:Conclusion}
We have derived the exact and asymptotic distributions of the traded volume, of the lower and upper bound of the clearing prices and of the clearing price range of a basic call auction model in which trading orders are submitted according to Poisson processes, with independent random prices with any general distribution, and a Markovian mechanism of cancellation.
The results obtained here generalize previous studies of the call auction in which the mathematical problem was only tackle through approximations.
Our work will hopefully raise the interest of the community in the need for more theoretical and empirical studies of the call auction. Firstly, the call auction is already widely used as a trading mechanism for the opening and closing of many financial markets, but has received until now considerably less attention than its continuous counterpart. Secondly, the recent regulation movement and the need to avoid dangerous phenomena directly linked to high-frequency trading strategies in a continuous double auction may 
be an opportunity for the call auction proponents to ascertain the potential benefits of this trading system. In any case, this calls for more empirical and theoretical works on the call auction.

\end{document}